\documentclass[11pt, a4paper]{article}
\usepackage{geometry}
\usepackage{amsthm}
\usepackage{amsmath,amssymb,amsfonts,mathtools}
\usepackage{comment}
\usepackage{bm}
\usepackage{enumerate}
\usepackage[dvipdfmx]{graphicx}
\usepackage[dvipdfmx]{color}
\usepackage{bbold}
\usepackage{here}
\usepackage{cite}
\usepackage{tikz}

\newtheorem{definition}{Definition}
\newtheorem{theorem}[definition]{Theorem}

\newtheorem{lemma}[definition]{Lemma}
\newtheorem{corollary}[definition]{Corollary}
\newtheorem{remark}[definition]{Remark}
\newtheorem{example}[definition]{Example}

\newcommand{\Br}{\mathbb{R}}
\newcommand{\Ora}[1]{\overrightarrow{#1}}

\newcommand{\Fp}[2][\sigma]{f_{#1(#2)}}
\newcommand{\Fplr}[3][\sigma]{\Fp[#1]{#2}\circ\cdots\circ\Fp[#1]{#3}}
\newcommand{\Flr}[2]{f_{#1}\circ\cdots\circ f_{#2}}
\newcommand{\Fall}[1][\sigma]{f^{#1}}

\newcommand{\Bp}[2][\sigma]{b_{#1(#2)}}

\newcommand{\Ap}[2][\sigma]{a_{#1(#2)}}
\newcommand{\Dp}[2][\sigma]{d_{#1(#2)}}

\newcommand{\Tf}[1]{\theta(f_{#1})}
\newcommand{\Tflr}[2]{\theta(\Flr{#1}{#2})}
\newcommand{\Tfp}[2][\sigma]{\theta(\Fp[#1]{#2})}

\newcommand{\Tfall}[1][\sigma]{\theta(\Fall[#1])}

\newcommand{\Mtt}[4]{\begin{pmatrix*} #1 & #2 \\ #3 & #4 \end{pmatrix*}}
\newcommand{\Mto}[2]{\begin{pmatrix*} #1 \\ #2 \end{pmatrix*}}
\newcommand{\Mot}[2]{\begin{pmatrix*} #1 & #2 \end{pmatrix*}}

\newcommand{\Mp}[2][\sigma]{M_{#1(#2)}}
\newcommand{\Mplr}[3][\sigma]{\Mp[#1]{#2}\cdots\Mp[#1]{#3}}

\newcommand{\Mall}[1][\sigma]{M^{#1}}

\usepackage{algorithmic}
\usepackage{algorithm}

\begin{document}

\title{Composition Orderings for Linear Functions and \\Matrix Multiplication Orderings}
\author{
Susumu Kubo
\and 
Kazuhisa Makino
\and 
Souta Sakamoto
}
\maketitle

\begin{abstract}
In this paper, we consider composition orderings for linear functions of one variable. 
Given $n$ linear functions $f_1,\dots,f_n:\Br \rightarrow\Br$ and a constant $c \in \Br$, the objective is to  find a permutation $\sigma:[n]\rightarrow[n]$ that minimizes/maximizes  $\Fplr{n}{1}(c)$, where $[n]=\{1, \dots , n\}$.
The problem is fundamental in many fields such as combinatorial optimization, computer science,  and operations research. 
It was first studied in the area of time-dependent scheduling, and known to be solvable in $O(n \log n)$ time if all functions are nondecreasing. 
In this paper, we present a complete characterization of optimal composition orderings for this case, by regarding linear functions as two-dimensional vectors. 
We also show several interesting properties on optimal composition orderings such as the equivalence between local and global optimality.
Furthermore, by using the characterization above, we provide a fixed-parameter tractable (FPT) algorithm for the composition ordering problem for general linear functions, with respect to the number of decreasing linear functions. 

We next deal with matrix multiplication orderings as a generalization of composition of linear functions. 
Given $n$ matrices $M_1,\dots,$ $M_n\in \Br^{m\times m}$ and two vectors $\bm{w},\bm{y}\in \Br^m$, where $m$ denotes a positive integer, the objective is to find a permutation $\sigma:[n]\rightarrow[n]$ that minimizes/maximizes $\bm{w}^\top \Mplr{n}{1} \bm{y}$.
The problem is also viewed as a generalization of  
flow shop scheduling through a limit. 
By extending the results for composition orderings for linear functions, 
we show that the multiplication ordering problem for $2\times 2$ matrices is solvable in $O(n \log n)$ time if all the matrices are  simultaneously triangularizable and 
have nonnegative determinants, 
and  FPT with respect to the number of matrices with negative determinants, 
 if all the matrices are simultaneously triangularizable.  
As the negative side, we finally prove that three possible natural generalizations are NP-hard: 1) when $m = 2$, even if  all the matrices have nonnegative determinants,  2) 
when $m \geq 3$, even if 
all the matrices are upper triangular with nonnegative elements, 
and 3) the target version of the problem, i.e., finding a permutation $\sigma$ with minimum $|\bm{w}^\top \Mplr{n}{1} \bm{y}-t|$ for a given target $ t\in \Br$, even if the problem corresponds to the composition ordering problem for monotone linear functions. 
\end{abstract}

\section{Introduction}\label{section:introduction}
In this paper, we consider composition ordering for linear functions, that is, polynomial functions of degree one or zero. 
Namely, given a constant $c \in \Br$ and $n$ linear functions $f_1,\dots,f_n:\Br \rightarrow\Br$, each of which is expressed as $f_i(x)=a_i x+b_i$ for some $a_i, b_i \in \Br$, we find a permutation $\sigma:[n]\rightarrow[n]$ that minimizes/maximizes  $\Fplr{n}{1}(c)$, where  $\Br$ denotes the set of  real numbers, and $[n] = \{1, \dots , n\}$ for a positive integer $n$.  
Since composition of functions is {\em not} commutative even for linear functions, i.e., $f_{\sigma(2)} \circ f_{\sigma(1)}\not=f_{\sigma(1)} \circ f_{\sigma(2)}$ holds in general, 
it makes sense to investigate the problem. 
For example, let $f_1(x)=-\frac{1}{2}x+\frac{3}{2}$, $f_2(x)=x-3$, $f_3(x)=3x-1$, and $c=0$,  
then the identity $\sigma$ (i.e., $\sigma(1)=1$, $\sigma(2)=2$ and $\sigma(3)=3$) provides 
$f_3\circ f_2\circ f_1(0) = f_3(f_2(f_1(0))) = f_3(f_2(\frac{3}{2})) = f_3(-\frac{3}{2}) = -\frac{11}{2}$, while the permutation $\tau$ with 
$\tau(1)=2$, $\tau(2)=1$ and $\tau(3)=3$ provides $f_3\circ f_1\circ f_2(0) = 8$ . In fact, we can see that $\sigma$ and $\tau$ are respectively minimum and maximum permutations for the problem. 
 The composition ordering problem is natural and fundamental
in many fields such as combinatorial optimization, computer science, and operations research. 

The problem was first studied from an algorithmic point of view under the name of 
{\em time-dependent scheduling}\cite{Gawiejnowicz,GuptaX2,Ho,Mosheiov,Tanaev,Wajs}.
We are given  $n$ jobs $j \in [n]$ with processing time $p_j$.  
Unlike the classical scheduling, 
the processing time $p_j$
is {\em not} constant, depending on the starting time of job $j$. 
 Here each $p_j$
is assumed to satisfy
$p_j(t) \leq  s + p_j(s+t)$ for any positive reals $s$ and $t$, 
since we should be able to finish processing job $j$ earlier if it starts earlier.
The model was introduced to deal
with learning and deteriorating effects. 
As the most fundamental setting of the time-dependent scheduling, 
we consider the linear model of  single-machine makespan minimization, 
where the makespan denotes the time when all the jobs have finished processing, and we assume that the
machine can handle only one job at a time and preemption is not allowed.
Linear model means that the processing time $p_j$ is  linear in the starting time $s$, i.e.,  $p_j(s) = \tilde{a}_j s + \tilde{b}_j$ for some constant $\tilde{a}_j$ and $\tilde{b}_j$. 
 Then it is not difficult to see that the model can be regarded as the minimum composition ordering problem for linear functions $f_j(x)=
 (\tilde{a}_j+1) x + \tilde{b}_j$, since $f_j$ represents the time to finish job $j$ if it start processing at time $x$. 
  Mosheiov \cite{Mosheiov} showed the makespan is independent of the schedule, i.e., any permutation provides the same composite, 
  if $\tilde{b}_j = 0$ for any $j \in [n]$. 
Gawiejnowicz and Pankowska \cite{Gawiejnowicz}, Gupta and Gupta \cite{GuptaX2}, Tanaev et al. \cite{Tanaev}, and Wajs \cite{Wajs} studied
the linear deterioration model, that is,  $\tilde{a}_j, \tilde{b}_j > 0$ (i.e., $a_j > 1$ and $b_j>0$) for any $j \in [n]$. 
Here $\tilde{a}_j$ and $\tilde{b}_j$ are
respectively called the {\em deterioration rate} and the {\em basic processing time} of job $j$. 
It can be shown that an optimal permutation can be obtained by arranging the jobs nonincreasingly with respect to $\tilde{a}_j/\tilde{b}_j\,(=(a_j-1)/b_j)$. 
Ho, Leung and Wei \cite{Ho} considered the linear shortening model, that is,  
$0> \tilde{a}_j >-1, \tilde{b}_j > 0$ (i.e., $1>a_j > 0$ and $b_j>0$) for any $j \in [n]$ and showed that an optimal permutation can be obtained again by arranging the jobs nonincreasingly with respect to $\tilde{a}_j/\tilde{b}_j\,(=(a_j-1)/b_j)$. 
Later,  
Kawase, Makino and Seimi \cite{KMS:linear} introduced the composition ordering problem,  and showed that the maximization can be formulated as the minimization problem, and 
propose an O($n \log n$)-time algorithm if all linear $f_i$'s are monotone nondecreasing, i.e.,  $a_i \geq 0$ for any $i \in [n]$. 
However, it is still open whether it is polynomially computable for general linear functions.  
Moreover, it is not known even when constantly many functions are monotone decreasing.

We remark that the time-dependent scheduling with 
the ready time and the deadline can be regarded as the composition ordering problem for piecewise linear functions, and is known to be NP-hard, and  Kawase, Makino and Seimi   \cite{KMS:linear} also studied the composition ordering for non-linear functions as well as the related problems such as partial composition and $k$-composition. 
We also remark that the free-order
secretary problem,  which is closely related to a branch of the problems such as the full-information secretary problem \cite{Ferguson},
knapsack and matroid secretary problems \cite{Babaioff1,Babaioff2,OveisGharan} and stochastic knapsack problems \cite{Dean05,Dean08},
can also be regarded as the composition ordering problem \cite{KMS:linear}.

\bigskip

 \subsection*{Main results obtained in this paper}

In this paper, we first characterize the minimum/maximum composition ordering for monotone (increasing) linear functions, in terms of their polar angles. 
In order to describe our result, we need to define three important concepts: counterclockwiseness, colinearity, and potential identity. 

We view a linear function $f(x)=ax+b$ as the vector $\Mto{b}{1-a}$ in $\Br^2$, and its {\em angle}, denoted by $\theta(f)$, is defined as the polar angle in $[0, 2\pi)$ of the vector, where we define $\theta(f)=\bot$ if the vector of $f$ is the origin $\Mto{0}{0}$. 
For example, if $\Ora{f}=\Mto{0}{2}$ and  $\Ora{g}=\Mto{-\sqrt{3}}{-1}$, then we have $\displaystyle \theta(f)=\frac{\pi}{2}$ and  $\displaystyle \theta(g)=\frac{7}{6}\pi$ (See also Figure \ref{figure:example_theta}). 
 \begin{figure}[htb]
     \centering
     \includegraphics[width=50mm]{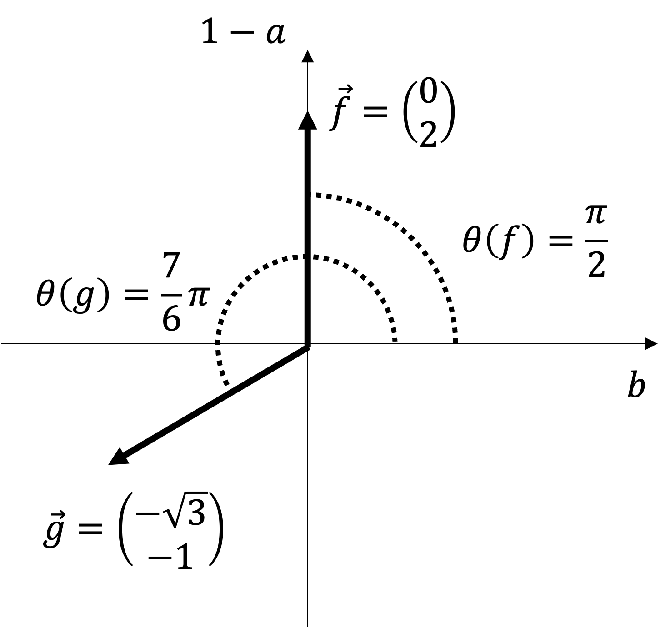}
     \caption{Vector representations of $f(x)=-x$ and $g(x)=2x-\sqrt{3}$.}
     \label{figure:example_theta}
 \end{figure}
For linear functions $f_1, \dots , f_n$, 
a permutation $\sigma:[n]\to [n]$ is called 
{\em counterclockwise}
if there exists an integer $k \in [n]$ such that $\theta(f_{\sigma(k)})\leq \dots \leq \theta(f_{\sigma(n)}) \leq \theta(f_{\sigma(1)}) \leq \dots \leq \theta(f_{\sigma(k-1)})$, where identical functions $f_i$ (i.e., $\theta(f_i)=\bot$) are ignored and the inequalities are assumed to be transitive. 
 Linear functions $f_1, \dots , f_n$ are called {\em colinear} if the corresponding vectors lie in some line through the origin,  
i.e., there exists an angle $\lambda$ such that $\theta(f_i) \in \{\lambda, \lambda+\pi, \bot\}$ for all $i \in [n]$, and {\em potentially identical} if there exists a counterclockwise permutation $\sigma:[n]\to [n]$ such that the corresponding composite 
is identical, i.e., $f_{\sigma(n)} \circ \dots \circ f_{\sigma(1)}(x)=x$. 
A permutation is called {\em minimum} (resp., {\em maximum}) if the corresponding composite is the minimum (resp., maximum). Then we have the following complete characterization of optimal permutations.

\begin{theorem}\label{theorem:main_theorem-x}
Let $f_1 , \dots , f_n$ be  monotone linear functions. Then we have the following statements.  
 \begin{description}
     \item[{\rm  (i)}] They are colinear if and only if any permutation 
     is minimum.  
      
     \item[{\rm  (ii)}] If they are not colinear, then  the following statements are   
     equivalent: 
      \begin{description}
     \item {\rm (ii-1)} They are potentially identical. 
     \item {\rm (ii-2)}
     A permutation  
     is minimum if and only if it is counterclockwise.  
     \end{description}
     \item[{\rm  (iii)}] If they are neither colinear nor potentially identical, then a permutation 
     is minimum if and only if it is a counterclockwise permutation such that  $\theta(\Fall)+\pi\in[\Tfp{t},\Tfp{s}]_{2\pi}$, where $s$ and $t$ denote the first and last integers $i$ such that $f_{\sigma(i)}$ is non-identical.
 \end{description}
    \end{theorem}
\noindent    
Here we define
 \begin{equation*}
     [\theta_1,\theta_2]_{2\pi}=
     \{
  \theta\in[\lambda_1,\lambda_2]
  \mid 
      \lambda_1=_{2\pi}\theta_1, 
   \ \lambda_2=_{2\pi}\theta_2, 
   \ \lambda_2-\lambda_1\in[0,2\pi) \}, 
 \end{equation*}
where 
 for two angles $\theta_1,\theta_2\in \Br$, we write $\theta_1=_{2\pi} \theta_2$ if they are congruent on the angle, i.e., $\theta_1-\theta_2 \in 2\pi \mathbb{Z}$, where $\mathbb{Z}$ denotes the set of integers. 
We note that the lexicographical orderings which Kawase et al. \cite{KMS:linear} introduced 
can be interpreted as counterclockwise permutations, and they showed the existence of counterclockwise minimum permutations. 

We also note that (i) Theorem \ref{theorem:main_theorem-x} can characterize maximum permutations by replacing ``counterclockwise"  by ``clockwise,"  which is obtained from 
the transformation between minimization and maximization in \cite{KMS:linear}  (See Section \ref{section:notation}), and (ii) Theorem \ref{theorem:main_theorem-x} can be generalized to characterize 
minimum/maximum permutations for monotone nondecreasing linear functions (See Section \ref{section:nondecreasing}). 

These results enable us to efficiently count and enumerate all minimum/maximum permutations. 
We also show that (i) a permutation is {\em globally} minimum (resp., maximum) if and only if it is {\em locally}\footnote{The term ``locally optimal'' will be defined in Section \ref{section:monotone}.} minimum (resp., maximum),  
and (ii) counterclockwise orderings are unimodal, which also reveals interesting discrete structures of composition orderings. 

We then deal with the composition ordering for general linear functions. 
We provide several structural properties of the optimal orderings.  
These, together with the characterization for monotone linear functions, provide a fixed-parameter tractable (FPT) algorithm for the composition ordering problem for general linear functions, 
with respect to the number of monotone decreasing linear functions. 

  \begin{theorem}\label{theorem:main3}
An optimal permutation for linear functions $f_1, \dots, f_n$ can be computed in $O(k 2^k  n^6)$ time, where $k\,(>0)$ denotes the number of monotone decreasing functions. 
\end{theorem}

\noindent
Theorem \ref{theorem:main3} means that both minimum and maximum permutations can be computed in $O(k 2^k  n^6)$ time. 
We remark that the FPT algorithm can be modified to 
efficiently count and enumerate all optimal permutations.

\medskip

In this paper, we finally consider the multiplication ordering for matrices as a generalization of the composition ordering for linear functions. 
The problem for matrices is to 
 find a permutation $\sigma:[n]\rightarrow[n]$ that minimizes/maximizes $\bm{w}^\top \Mplr{n}{1} \bm{y}$ for given  $n$ matrices $M_1,\dots,$ $M_n\in \Br^{m\times m}$ and two vectors $\bm{w},\bm{y}\in \Br^m$, 
 where $m$ denotes a positive integer. 
In fact, if $\bm{w}=\Mto{1}{0}
$,  $\bm{y}=\Mto{0}{1}$, and $M_i=\Mtt{a_i}{b_i}{0}{1}$ for any $i \in [n]$, 
then the matrix multiplication ordering problem is equivalent to the composition ordering problem for linear functions $f_i(x) = a_i x+ b_i$.
Here, we consider the minimization problem, since the maximization problem for $(M_1, \ldots, M_n, \bm{w}, \bm{y})$ corresponds to the minimization one for $(M_1, \ldots, M_n, -\bm{w}, \bm{y})$.

We obtain the following generalization of the results of linear functions. 
Matrices $M_1, \dots,$ $ M_n\in \Br^{m\times m}$ are called {\em simultaneously triangularizable} if  there exists a regular matrix $P\in\Br^{m\times m}$ such that $P^{-1}M_iP$ is an upper triangular matrix for any $i\in[n]$.


\begin{theorem}\label{theorem:triangularizable}
     For the optimal multiplication ordering problem for $2\times 2$ simultaneously triangularizable matrices, the following statements hold.
     \begin{enumerate}
  \item[{\rm (i)}] If all matrices have nonnegative  determinants, then an optimal multiplication ordering can be computed in $O(n\log n)$ time.
  \item[{\rm (ii)}] If some matrix has a negative determinant, then an optimal multiplication ordering can be computed in $O(k2^k n^6)$ time, where $k$ denotes the number of matrices with negative determinants. 
     \end{enumerate}
 \end{theorem}

We remark that Theorem \ref{theorem:triangularizable} (i) can be regarded as a refinement of the result by  
Bouquard, Lenté and Billaut \cite{MaxPlus} for multiplication ordering of $2 \times 2$ upper triangular  matrices in the max-plus algebra, 
since the objective value in the max-plus algebra is obtained by taking the limit of the one in linear algebra. 
They showed that the problem in the max-plus algebra is a  generalization of the two-machine flow shop scheduling problem to minimize the makespan, and solvable in $O(n\log n)$ time by using an extension of  Johnson's rule \cite{Johnson} for the two-machine flow shop scheduling.  
Kubo and Nishinari examined the relationship between the flow shop scheduling and usual matrix multiplication \cite{Kubo}. 
We obtain the result in \cite{MaxPlus} as a corollary of our result in Section \ref{section:matrix}, . 


As a negative side, we show that all possible natural generalizations turn out to be intractable unless P$=$NP.

 \begin{theorem}\label{theorem:matrix-hardness1}
 \begin{enumerate}
 \item[{\rm (i)}]
   It is strongly NP-hard to compute an optimal multiplication ordering for $2\times 2$ matrices, even if all matrices have nonnegative entries and determinants.
 \item[{\rm (ii)}]
      It is strongly NP-hard to compute an optimal multiplication ordering for $m\times m$ matrices with $m \geq 3$, even if all matrices are nonnegative (i.e., all the entries are nonnegative) and upper triangular. 
\end{enumerate}
 \end{theorem}

We also deal with the target version of the multiplication ordering problem for matrices, i.e., minimizing the value $|\bm{w}^\top \Mplr{n}{1} \bm{y}-t|$ for a given target $t \in \Br$. 
Unfortunately, it is also strongly NP-hard.

 \begin{theorem}\label{theorem:minimize_norm-hardness}
Given a target $t \in \Br$, it is strongly NP-hard to compute an permutation $\sigma:[n]\to [n]$ such that 
$|\bm{w}^\top \Mplr{n}{1} \bm{y}-t| \leq c_1  \cdot \min_\rho |\bm{w}^\top \Mplr[\rho]{n}{1} \bm{y}-t|+c_2$ for any positive $c_1$ and $c_2$, even if the problem corresponds to composition ordering for monotone linear functions. 
 \end{theorem}

    \subsection*{The organization of the paper}
 The rest of the paper is organized as follows.
 Section \ref{section:notation} provides some notation and basic properties needed in the paper. 
 In Section \ref{section:monotone},  we consider composition orderings for monotone linear functions. 
 We provide the proof of Theorem \ref{theorem:main_theorem-x} and structural properties such as the equivalence between local and global optimality.
  In Section \ref{section:nondecreasing}, we show that the nondecreasing case 
  is reduced to the monotone case and characterize optimal composition orderings. 
  Section \ref{section:general} proves Theorem \ref{theorem:main3} by making an FPT algorithm of the general case. 
In Section \ref{section:matrix} we generalize composition of linear functions to matrix multiplication in linear and max-plus algebras, 
 and prove Theorems \ref{theorem:triangularizable}, \ref{theorem:matrix-hardness1}, and \ref{theorem:minimize_norm-hardness}. 

\section{Notation and Basic Properties} \label{section:notation}
 In this section, we first fix notation
and state several basic properties of linear functions,
which will be used in this paper. 
 We then mention that minimum and maximum compositions are polynomially equivalent. 

We view a linear function $f(x)=ax+b$ as the vector $\Mto{b}{1-a}$ in $\Br^2$, and its {\em angle}, denoted by $\theta(f)$, is defined as the polar angle in $[0, 2\pi)$ of the vector, where we define $\theta(f)=\bot$ if the vector of $f$ is the origin $\Mto{0}{0}$. 

 For two reals $\ell$ and $r$ with $\ell < r$, let  $[\ell,r] =\{x\in\Br\mid l\leq x \leq r\}$.  
 Similarly, we denote semi-open intervals by $(\ell,r]$ and $[\ell,r)$ and open intervals by $(\ell,r)$.
For a linear function $f(x)=ax+b$, 
 we respectively denote by $\alpha(f)$ and $\beta(f)$ the slope and intercept of $f(x)$, i.e., $\alpha(f)=a$ and $\beta(f)=b$.
A linear function $f$ is respectively called {\em monotone $($increasing$)$}, {\em constant},  and {\em monotone decreasing} if $\alpha(f)>0$, $\alpha(f)=0$, and $\alpha(f)<0$. 
Since the result of arithmetic operations on angles may take a value outside of $[0,2\pi)$, 
 we provide some notation to deal with such situations, some of which have already been used in the introduction. 
  For two angles $\theta_1,\theta_2\in \Br$, we write $\theta_1=_{2\pi} \theta_2$ if they are congruent on the angle, i.e., $\theta_1-\theta_2 \in 2\pi \mathbb{Z}$, 
  and define
 \begin{equation*}
     [\theta_1,\theta_2]_{2\pi}=
     \{
  \theta\in[\lambda_1,\lambda_2]
  \mid 
      \lambda_1=_{2\pi}\theta_1, 
   \ \lambda_2=_{2\pi}\theta_2, 
   \ \lambda_2-\lambda_1\in[0,2\pi) \}.
 \end{equation*}
 For example, if $\theta_1=\frac{\pi}{6}$ and $\theta_2=\frac{2\pi}{3}$ then 
 \[[\theta_1,\theta_2]_{2\pi}=\cdots\cup\left[-\frac{11\pi}{6},-\frac{4\pi}{3}\right]\cup\left[\frac{\pi}{6},\frac{2\pi}{3}\right]\cup\left[\frac{13\pi}{6},\frac{8\pi}{3}\right]\cup\cdots.
 \]
 We similarly define  open and semi-open intervals such as $(\theta_1,\theta_2)_{2\pi}$, $[\theta_1,\theta_2)_{2\pi}$, and $(\theta_1,\theta_2]_{2\pi}$.
 For a non-interval set $S$, we define $S_{2\pi}=\{\theta \mid \theta =_{2\pi}\lambda\text{ for }\lambda\in S\}$.
 


    
We next state four basic properties of linear functions.  Note that Lemmas \ref{lemma:theta_composition=x1}, \ref{lemma:sin}, and \ref{lemma:vector_composition}  
 do not assume the monotonicity of linear functions.

\begin{lemma}\label{lemma:theta_composition=x1}
Let $g$ be the identical function, i.e., $g(x)=x$. 
Then for any function $h$, we have 
 $h\circ g=g\circ h=h$.
\end{lemma}
\begin{proof}
Straightforward.
\end{proof}

 \begin{lemma}\label{lemma:sin}
For two non-identical linear functions $g$ and $h$, we have the following two equivalences. 
\begin{description}
  \item[{\rm (i) }]
      $ h\circ g <\footnote{The inequality for functions means that the inequality holds for any argument.} g\circ h \  \Leftrightarrow \   \theta(h)-\theta(g) \in (0,\pi)_{2\pi}$. 
  \item[{\rm (ii)}]
      $h\circ g = g\circ h \ \Leftrightarrow \  \theta(h)-\theta(g) \in \{0,\pi\}_{2\pi}$.
     \end{description}
 \end{lemma}
 \begin{proof}
     It follows from the equalities
     \begin{align*}
  g\circ h(x)-h\circ g(x)
  &=\bigl(\alpha(g)(\alpha(h)x+\beta(h))+\beta(g)\bigr) -\bigl(\alpha(h)(\alpha(g)x+\beta(g))+\beta(h)\bigr)\\
  &=\beta(g)(1-\alpha(h))-\beta(h)(1-\alpha(g))\\
  &=|\vec{g}||\vec{h}|\sin(\theta(h)-\theta(g)). 
     \end{align*}
 \end{proof}
 
 \begin{lemma}\label{lemma:vector_composition}
     Let $g$ and $h$ be two linear functions, then $\Ora{h\circ g} = \vec{h}+\alpha(h)\vec{g}$.
 \end{lemma}
 \begin{proof}
     It follows from
     \begin{align*}
  \Ora{h\circ g} 
  &= \Mto{\alpha(h)\beta(g)+\beta(h)}{1-\alpha(h)\alpha(g)} 
  = \Mto{\beta(h)}{1-\alpha(h)} + \alpha(h)\Mto{\beta(g)}{1-\alpha(g)}
  = \vec{h}+\alpha(h)\vec{g}. 
     \end{align*}
 \end{proof}

\begin{lemma}\label{lemma:theta_composition}
     For non-identical monotone linear functions $g$ and $h$, we have the following statements.
     
\begin{description}
\item[{\rm (i)}] $\theta(h) - \theta(g) \in (0,\pi)_{2\pi} \, \Leftrightarrow \, \theta(h\circ g) \in (\theta(g),\theta(h))_{2\pi} \, \Leftrightarrow \, \theta(g\circ h) \in (\theta(g),\theta(h))_{2\pi} $. 

\item[{\rm (ii)}] $\theta(h) - \theta(g) \in \{0,\pi\}_{2\pi} \,  \Leftrightarrow \, \theta(h\circ g) \in \{\theta(g),\theta(h),\bot\} \Leftrightarrow \, \theta(g\circ h) \in \{\theta(g),\theta(h),\bot\}$.

\item[{\rm (iii)}] $\theta(h) = \theta(g) \, \Rightarrow \,\theta(h\circ g) =\theta(g\circ h) =\theta(h)\,(=\theta(g)).$

\item[{\rm (iv)}]  $\theta(h\circ g) =\bot   \, \Leftrightarrow \, \theta(g\circ h) =\bot  \, \Rightarrow  \, \theta(h) - \theta(g)  =_{2\pi} \pi$.
\end{description}
\end{lemma}
 
\begin{proof}
(i), (ii), and (iii) follow from  Lemma \ref{lemma:vector_composition}, 
$\Ora{h}, \Ora{g}\not=0$, and $\alpha(h)>0$.
(iv) follows from (ii), (iii), and Lemma \ref{lemma:sin} (ii). \end{proof}  

In fact, the equivalence of (iv) holds for general linear functions $g$ and $h$.

For linear functions $f_1, \dots , f_n$ and a permutation $\sigma:[n]\to [n]$, we denote $f_{\sigma(n)} \circ \dots \circ f_{\sigma(1)}$ by $f^\sigma$.

Before ending this section, we provide a linear-time transformation between the maximization problem and the minimization problem \cite{KMS:linear}. 
 For a linear function $f(x)=ax+b$, we define a linear function $\tilde{f}$ by
 \begin{gather}
     \tilde{f}(x) \ = \ ax-b.
\label{eq--0001}
 \end{gather}
Note that $\widetilde{f}$ is monotone if $f$ is monotone. 
For linear functions $f_1, \dots , f_n$ and a permutation $\sigma:[n]\to [n]$, we have $\beta(f^\sigma) \ = \ -\beta(\tilde{f}^\sigma)$.
Since any permutation $\sigma:[n]\to [n]$ provides $\alpha(f^\sigma)=\prod_{i \in [n]}\alpha(f_i)$, we can see that 
 the maximum composition for $f_1,\dots,f_n$ is equivalent to the minimum composition  for $\tilde{f}_1,\dots,\tilde{f}_n$. 
 where this transformation was mentioned in \cite{KMS:linear}.
 Therefore, we mainly deal with the minimization problem for linear functions, and sometimes use the term ``optimal" instead of ``minimum".

\section{Composition of Monotone  Linear Functions}\label{section:monotone}
In this section, we consider composition orderings for monotone linear functions. 
Especially, we prove Theorem \ref{theorem:main_theorem-x} and show 
 structural properties of the composites.  
\subsection{Proof of Theorem \ref{theorem:main_theorem-x} (i)}
We first prove Theorem \ref{theorem:main_theorem-x} (i), which can be easily obtained from basic properties in Section \ref{section:notation}. 

\begin{proof}[Proof of Theorem \ref{theorem:main_theorem-x}(i)]
Let us first show the only-if part. 
For any $i\in[n-1]$, let $\rho_i:[n]\to [n]$ be the $i$-th adjacent transposition, i.e., the transposition of two consecutive integers $i$ and $i+1$. 
Let ${\rm id}:[n]\to [n]$ denote the identity permutation.
Then we have $f^{\rho_i} = f^{\rm id}$, 
since $f_{i} \circ f_{i+1} =f_{i+1} \circ f_{i} $
by Lemmas \ref{lemma:theta_composition=x1}
and \ref{lemma:sin} (ii).
It is well-known that any permutation can be obtained by a product of adjacent transpositions and therefore for any permutation $\sigma$ we obtain $f^\sigma = f^{\rm id}$, which is minimum.

For the if part, suppose, without loss of generality, that $f_1$ and $f_2$ are not colinear.
Then we have $f_1 \circ f_2 \not= f_2 \circ f_1$ by Lemma \ref{lemma:sin} (i), which implies that 
$f_1 \circ f_2 \circ (f_n \circ \dots \circ f_3) \ \not= \  f_2 \circ f_1 \circ (f_n \circ \dots \circ f_3)$, 
which completes the proof of the if part. 
\end{proof}

Note that in fact Theorem \ref{theorem:main_theorem-x} (i) does not require the monotonicity, and hence it is true even if $f_i$'s are general linear functions.

\subsection{Proof of Theorem \ref{theorem:main_theorem-x} (ii)}
In order to prove Theorem \ref{theorem:main_theorem-x} (ii), 
we first define the neighborhood of a permutation and the local optimality of permutations. 
     Let $\sigma :[n]\rightarrow[n]$ be a permutation.
     For three positive integers $\ell$, $m$ and $r$ with $\ell\leq m< r$, define a permutation $\sigma_{\ell,m,r}:[n]\rightarrow[n]$ by
     \begin{align*}
  \sigma_{\ell,m,r}(i)&=
  \begin{cases}
      \sigma(i) & (1\leq i<\ell, r<i\leq n),\\
      \sigma(i-\ell+m+1) & (\ell\leq i<\ell-m+r),\\
      \sigma(i+m-r) & (\ell-m+r\leq i\leq r)
  \end{cases}
  \end{align*}
which is illustrated in Fig. \ref{figure:neighbor}. In particular, $\sigma_{1,k,n}$ is abbreviated as $\sigma_k$ and referred to as the \textit{$k$-shift} of $\sigma$. 
 \begin{figure}[htb]
     \centering
     \includegraphics[width=120mm]{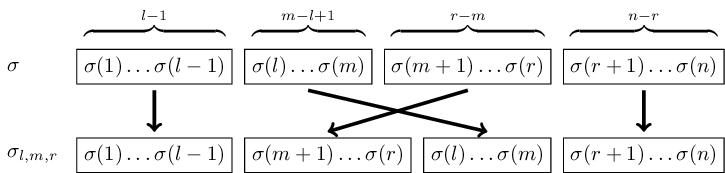}
     \caption{$Permutation \sigma_{l,m,r}$ obtained from $\sigma$ by swapping two adjacent intervals.}
     \label{figure:neighbor}
 \end{figure}
The \textit{neighborhood} $N(\sigma)$ of $\sigma$ is defined by
$N(\sigma)=\{\sigma_{\ell,m,r}\mid \ell\leq m< r\}$, 
that is, the set of permutations obtained from $\sigma$ by swapping two adjacent intervals in $\sigma$.
A permutation $\sigma$ is \textit{locally optimal} if $\Fall \leq \Fall[\mu]$ for any permutation $\mu\in N(\sigma)$.
 
 The next lemma plays an important role throughout the paper and its proof is technical and involved. 
 \begin{lemma}\label{lemma:localopt_counterclockwise}
 A locally optimal permutation for non-colinear monotone linear functions is counterclockwise.
 \end{lemma}

 For the proof, we provide the following four lemmas.
The first lemma directly follows from the definition and Lemma \ref{lemma:sin}. 
 \begin{lemma}\label{lemma-mx1}
For monotone linear functions $f_1, \dots , f_n$, a permutation $\sigma:[n]\to [n]$ is  locally optimal if and only if $\theta(f_{\sigma(r)}\circ \cdots \circ f_{\sigma(m+1)})-\theta(f_{\sigma(m)}\circ \cdots \circ f_{\sigma(\ell)}) \in [0,\pi]_{2\pi}$ holds for any integers $\ell, m$ and $r$ with $\ell \leq m< r$ such that $f_{\sigma(r)}\circ \cdots \circ f_{\sigma(m+1)}$ and $f_{\sigma(m)}\circ \cdots \circ f_{\sigma(\ell)}$ are non-identical. 
 \end{lemma}
 \begin{proof}
 By definition,  $\sigma$ is locally optimal if and only if $f^\sigma \leq f^{\sigma_{\ell,m,r}}$ for any integers $\ell, m$ and $r$ with $\ell \leq m< r$, 
 which is equivalent to that 
$f_{\sigma(r)} \circ \dots \circ f_{\sigma(\ell)} \leq 
(f_{\sigma(m)}\circ \cdots \circ f_{\sigma(\ell)}) \circ (f_{\sigma(r)}\circ \cdots \circ f_{\sigma(m+1)})$.  
By Lemma \ref{lemma:sin},  this is further equivalent to the condition in the lemma. 
\end{proof}

We then show three useful properties of locally optimal permutations. 
 \begin{lemma}\label{lemma:no_pi}
    For monotone linear functions $f_1, \dots , f_n$, let $\sigma:[n]\to [n]$ be  locally optimal.
If $\Tfp{j}-\Tfp{i} =_{2\pi} \pi$ for some $i, j\in[n]$ with $i < j$ and  $\Tfp{\ell}=\bot$ for all $\ell \in [n]$ with $i < \ell < j$,  then $f_1, \dots , f_n$ are colinear.
 \end{lemma}
 \begin{proof}
     Without loss of generality, we assume that $\sigma$ is the identity permutation.
We first prove the statement when no $f_i$ for $i\in [n]$ is identical and $\Tf{i+1}\neq \Tf{i}$ holds for all $i\in[n-1]$.
Let $j$ be a positive integer with  $\Tf{j+1}-\Tf{j} =_{2\pi} \pi$.  
We claim that  $\Tf{i+1}-\Tf{i} =_{2\pi} \pi$ for any  $i\in [n-1]$. 

Lemma \ref{lemma-mx1}  implies that $\Tf{j+2} \in (\Tf{j+1},\Tf{j+1}+\pi]_{2\pi}$, since $\sigma$ is locally optimal.  
If $\Tf{j+2} -\Tf{j+1}\not=_{2\pi} \pi$, then we have 
\begin{align*}
  \theta(f_{j+2}\circ f_{j+1})\in(\Tf{j+1},\Tf{j+2})_{2\pi}\subseteq(\Tf{j+1},\Tf{j+1}+\pi)_{2\pi}, 
\end{align*}
implying that 
\begin{align*}
  \theta(f_{j+2}\circ f_{j+1})-\Tf{j}\in(\pi,2\pi)_{2\pi}.
\end{align*}
     This contradicts the local optimality of $\sigma$ by Lemma \ref{lemma-mx1}.
Therefore, we have $\Tf{j+2}-\Tf{j+1} =_{2\pi} \pi$, and by repeatedly applying this argument, we obtain $\Tf{i+1}-\Tf{i} =_{2\pi} \pi$ for all $i \geq j$. 
The case of $i$ with $i \leq j-1$ is treated similarly. 
In fact,  Lemma \ref{lemma-mx1}, together with the local optimality of $\sigma$, implies that $\Tf{j} \in (\Tf{j-1},\Tf{j-1}+\pi]_{2\pi}$.
If  $\Tf{j} -\Tf{j-1}\not=_{2\pi} \pi$, then we have 
\begin{align*}
\Tf{j+1} - \theta(f_{j}\circ f_{j-1})\in(\pi,2\pi)_{2\pi}, 
\end{align*}
which contradicts the local optimality of $\sigma$ by Lemma \ref{lemma-mx1}.
Again by repeatedly applying this argument, we obtain $\Tf{i+1}-\Tf{i} =_{2\pi} \pi$ for all $i \leq j$. 
Therefore, $\Tf{i+1}-\Tf{i}=_{2\pi}\pi$ holds for any $i\in[n-1]$. 

We now turn to the general case.
Let us partition $[n]$ into fewer intervals $S_1, \dots, S_k$ such that $|\theta(S_p)\setminus \{\bot\}|=1$ for all $p \in [k]$, where $\theta(S_p)=\{\theta(f_i) \mid i \in S_p\}$.  
Namely, there exist $(k+1)$ positive integers $\ell_1=1< \ell_2 <\cdots<\ell_k <  \ell_{k+1} = n$ and reals $\lambda_1 \not= \lambda_2 \not= \dots  \not=\lambda_k$ such that $S_p=\{j \in \mathbb{Z}\mid \ell_p \leq  j < \ell_{p+1}\}$ and $\theta(S_p)\setminus \{\bot\}=\{\lambda_p\}$ for all $p \in [k]$.  
For an index $p \in [k]$, let $g_p=\Flr{\ell_{p+1}-1}{\ell_p}$.
Then it is not difficult to see that  $g_p$'s are all monotone linear, the ordering $(g_1,\dots g_k)$ is locally optimal, and  $\theta(g_{p+1})-\theta(g_{p}) =_{2\pi} \pi$ for some $p\in[k-1]$. 
Thus by applying the proof for the first case, we obtain that $g_1, \dots , g_k$ are colinear, which implies that $f_1, \dots , f_n$ are colinear. 
 \end{proof}


\begin{lemma}\label{lemma-cor:no_x}
For monotone linear functions $f_1, \dots , f_n$,  let $\sigma: [n]\to [n]$ be a local optimal permutation.
If an interval $S=\{i \in \mathbb{Z} \mid \ell \leq i \leq r \}$  contains a non-identical function $f_{\sigma(i)}$ 
 and satisfies $\Fplr{r}{\ell}(x)=x$, then 
 $f_{\sigma(j)}$ is identical for any $j \in [n]\setminus S$, 
 unless $f_1, \dots , f_n$ are colinear.
 \end{lemma}
 \begin{proof}
 Assuming that there exists an interval $S$ which satisfies as above and a $j\in[n]\setminus S$ such that $f_{\sigma(j)}(x)\neq x$, we show that $f_1,\dots,f_n$ are colinear.
Without loss of generality, we  assume that $\sigma$ is the identity permutation. Let $S$ be a minimal interval that satisfies the condition in the lemma. 
By the condition and the minimality of $S$, Lemma \ref{lemma:theta_composition} (iv) implies that any index $m \in S \setminus \{r\}$ satisfies  $\theta(\Flr{r}{m+1})- \theta(\Flr{m}{\ell})=_{2\pi}\pi$. Note that an ordering $(f_1, \dots , f_{\ell-1}, \Flr{m}{\ell}, \Flr{r}{m+1}, f_{r+1}, \dots, f_n)$ is locally optimal for the functions.  
Thus the linear functions corresponding to 
\[
\{\Tf{i} \mid i\in [n]\setminus S\} \cup \{\Tflr{r}{m+1}, \Tflr{m}{\ell} \mid m \in S \setminus \{r\}\}
\]
are colinear by Lemma \ref{lemma:no_pi}.
In particular, 
\[\{\Tf{i} \mid i\in [n]\setminus S\} \cup \{\Tf{r}\}\cup \{\Tflr{m}{\ell} \mid m \in S \setminus \{r\}\}\]
are colinear.
Then $f_1,\dots,f_n$ are colinear by Lemma \ref{lemma:theta_composition} (ii).
\end{proof}

For monotone non-identical linear functions $f_1, \dots , f_n$, 
let us assume without loss of generality that the identity permutation  ${\rm id}:[n]\to [n]$ is locally optimal. 
Let $\lambda_1, \dots,\lambda_n$ denote  the smallest reals such that 
\begin{equation}
\label{eq-new-mx1}
\begin{array}{l} \lambda_i=_{2\pi}\Tf{i}\  \mbox{ for all } \ i \in [n]\\[.15cm]
 0\leq\lambda_1\leq\lambda_2\leq\cdots\leq\lambda_n. 
     \end{array}
     \end{equation}
Note that they are well-defined. 

\begin{lemma}
\label{lamma-newm001}
Let $\lambda_1, \dots , \lambda_n$ be defined as above.
If $\ell$ and $r$ in $[n]$ satisfy $\lambda_r-\lambda_\ell<2\pi$ and $\ell \leq r$, then 
we have $\theta(f_r \circ \dots \circ f_\ell) \in [\theta(f_\ell),\theta(f_r)]_{2\pi}\cup\{\bot\}$.
In addition, if $f_1,\dots,f_n$ are not colinear and $\lambda_\ell<\lambda_r$ then 
we have $\theta(f_r \circ \dots \circ f_\ell) \in (\theta(f_\ell),\theta(f_r))_{2\pi}\cup\{\bot\}$.
\end{lemma}
\begin{proof}
Note that the identity permutation is locally optimal.
By Lemma \ref{lemma:theta_composition} (iii), we may assume without loss of generality that $\lambda_i<\lambda_{i+1}$ for any $i\in[n-1]$.
Initially, we show the first claim.
If $\ell=r$ then the claim clearly holds and thus we assume $\ell<r$.
By Lemmas \ref{lemma:theta_composition} (i)\,(ii) and \ref{lemma-mx1},   we have 
$\theta(f_{\ell+1} \circ f_\ell) \in [\theta(f_\ell),\theta(f_{\ell+1})]_{2\pi}\cup\{\bot\}$. 
If $\ell+1=r$, then it proves the statement of the first claim.
On the other hand, if $\ell+1 < r$, then we divide our discussion into two cases, whether $\theta(f_{\ell+1} \circ f_\ell)\in[\theta(f_\ell),\theta(f_{\ell+1})]_{2\pi}$ or not.
If $\theta(f_{\ell+1} \circ f_\ell)\in[\theta(f_\ell),\theta(f_{\ell+1})]_{2\pi}$ then we note that an ordering $( f_{1}, \dots , f_{\ell-1}, f_{\ell+1}\circ f_\ell, f_{\ell+2}, \dots , f_n)$ is locally optimal for linear functions $f_{1}, \dots , f_{\ell-1}, f_{\ell+1}\circ f_\ell, f_{\ell+2}, \dots , f_n$, and $f_{\ell+1} \circ f_\ell, f_{\ell+2}, \dots , f_r$ satisfies the condition in the lemma. 
Thus by repeatedly applying the argument above to $f_{\ell+1}\circ f_\ell, f_{\ell+2}, \dots , f_r$, we can obtain the proof of the first claim.
If $\theta(f_{\ell+1} \circ f_\ell)=\bot$ then we note that an ordering $( f_{1}, \dots , f_{\ell-1}, f_{\ell+2}, \dots , f_n)$ is locally optimal for linear functions $f_{1}, \dots , f_{\ell-1}, f_{\ell+2}, \dots , f_n$, and $f_{\ell+2}, \dots , f_r$ satisfies the condition in the lemma. 
Thus by repeatedly applying the argument above to $f_{\ell+2}, \dots , f_r$, we can obtain the proof of the first claim.

Next, we prove the second claim by showing that $\theta(f_r\circ\cdots\circ f_\ell)\in(\theta(f_\ell),\theta(f_r)]_{2\pi}\cup\{\bot\}$ and $\theta(f_r\circ\cdots\circ f_\ell)\in[\theta(f_\ell),\theta(f_r))_{2\pi}\cup\{\bot\}$.
By Lemmas \ref{lemma-mx1} and \ref{lemma:no_pi}, $\theta(f_{i+1})-\theta(f_i)\in(0,\pi)_{2\pi}$ for any $i\in[n-1]$.

We first show $\theta(f_r\circ\cdots\circ f_\ell)\in(\theta(f_\ell),\theta(f_r)]_{2\pi}\cup\{\bot\}$.
By Lemma \ref{lemma:theta_composition} (i), we have $\theta(f_{\ell+1}\circ f_\ell)\in(\theta(f_\ell),\theta(f_{\ell+1}))_{2\pi}$.
By the first claim of lemma, $\theta(f_r\circ\cdots\circ f_{\ell+2})\in[\theta(f_{\ell+2}),\theta(f_r)]_{2\pi}\cup\{\bot\}$.
If $\theta(f_r\circ\cdots\circ f_{\ell+2})=\bot$ then it proves the statement of the second claim because $\theta(f_r\circ\cdots\circ f_\ell)=\theta(f_{\ell+1}\circ f_\ell)\in(\theta(f_\ell),\theta(f_{\ell+1}))_{2\pi}$.
On the other hand, if $\theta(f_r\circ\cdots\circ f_{\ell+2})\in[\theta(f_{\ell+2}),\theta(f_r)]_{2\pi}$, then $\theta(f_r\circ\cdots\circ f_{\ell+2})-\theta(f_{\ell+1}\circ f_\ell)\in[0,\pi]_{2\pi}$ by Lemma \ref{lemma-mx1}.
Therefore, we have $\theta(f_r\circ\cdots\circ f_\ell)\in[\theta(f_{\ell+1}\circ f_\ell),\theta(f_r\circ\cdots\circ f_{\ell+2})]_{2\pi}\cup\{\bot\}\subseteq(\theta(f_l),\theta(f_r)]_{2\pi}\cup\{\bot\}$, where the first relationship is followed by Lemma \ref{lemma:theta_composition} (i).

Similarly, we show $\theta(f_r\circ\cdots\circ f_\ell)\in[\theta(f_\ell),\theta(f_r))_{2\pi}\cup\{\bot\}$.
By Lemma \ref{lemma:theta_composition} (i), $\theta(f_r\circ f_{r-1})\in(\theta(f_{r-1}),\theta(f_r))_{2\pi}$.
By the first claim of lemma, $\theta(f_{r-2}\circ\cdots\circ f_\ell)\in[\theta(f_\ell),\theta(f_{r-2})]_{2\pi}\cup\{\bot\}$.
If $\theta(f_{r-2}\circ\cdots\circ f_\ell)=\bot$ then it proves the statement of the second claim because $\theta(f_r\circ\cdots\circ f_\ell)=\theta(f_r\circ f_{r-1})\in(\theta(f_{r-1}),\theta(f_r))_{2\pi}$.
On the other hand, if $\theta(f_{r-2}\circ\cdots\circ f_\ell)\in[\theta(f_\ell),\theta(f_{r-2})]_{2\pi}$ then $\theta(f_r\circ f_{r-1})-\theta(f_{r-2}\circ\cdots\circ f_{\ell})\in[0,\pi]_{2\pi}$ by Lemma \ref{lemma-mx1}.
Therefore, we have $\theta(f_r\circ\cdots\circ f_\ell)\in[\theta(f_{r-2}\circ\cdots\circ f_{\ell}),\theta(f_r\circ f_{r-1})]_{2\pi}\cup\{\bot\}\subseteq[\theta(f_l),\theta(f_r))_{2\pi}\cup\{\bot\}$, where the first relationship is followed by Lemma \ref{lemma:theta_composition} (i).
\end{proof}

\begin{proof}[Proof of Lemma \ref{lemma:localopt_counterclockwise}]
Without loss of generality, we assume that all $f_i$’s are non-identical and 
the locally optimal permutation is the identity permutation.
Let $\lambda_1, \dots,\lambda_n$ denote  the reals defined before Lemma \ref{lamma-newm001}. 
Note that 
the identity permutation is counterclockwise if and only if  $\lambda_n\leq\lambda_1+2\pi$.
      Suppose to the contrary that $\lambda_n>\lambda_1+2\pi$, and define  $u$ and $v$ by 
\begin{equation}
\label{eq-wx-1}
\begin{array}{lll}
  u&=\min\{i\in[n]\mid\lambda_i>\lambda_1+2\pi\},\\[.1cm]
  v&=\max\{i\in[n]\mid\lambda_i\leq\lambda_u-2\pi\}.
\end{array}
\end{equation}
     Then by definition, we have $\lambda_{v+1}>\lambda_u-2\pi$.
Since $\sigma$ is locally optimal,  it follows from Lemmas \ref{lemma-mx1} and \ref{lemma:no_pi}  that 
\begin{equation}
\label{eq-xnew001}
\lambda_{v+1}<\lambda_v+\pi\leq\lambda_u-\pi<\lambda_{u-1}.  
\end{equation}
which in particular implies $v+1<u-1$.
We define $\varphi_1$ and $\varphi_2$ by 
\begin{align*}
  \varphi_1&=\Tflr{v}{1},\\
  \varphi_2&=\Tflr{u-1}{v+1}. 
\end{align*}
Since $\Tf{1},\dots,\Tf{v}\in[\Tf{u-1},\Tf{u}]_{2\pi}$ and $\Tf{1}\in[\Tf{u-1},\Tf{u})_{2\pi}$ by \eqref{eq-wx-1} and $\Tf{u} - \Tf{u-1} \in (0, \pi)_{2\pi}$
by Lemmas \ref{lemma-mx1} and \ref{lemma:no_pi}, 
Lemma \ref{lemma:theta_composition}  implies that
 \begin{equation}
 \label{eq-0}
  \varphi_1\in[\Tf{u-1},\Tf{u})_{2\pi}. 
\end{equation}
Note that $0< \lambda_{u-1}-\lambda_{v+1} < \lambda_{u-1}-\lambda_1 \leq 2\pi$ holds by \eqref{eq-wx-1} and \eqref{eq-xnew001}, and thus it follows from Lemmas \ref{lemma-cor:no_x} and \ref{lamma-newm001} that 
\begin{equation}
\label{eq-1}
\varphi_2\in(\Tf{v+1},\Tf{u-1})_{2\pi}.
\end{equation}     
Furthermore, by the local optimality of $\sigma$,  Lemma \ref{lemma-mx1} implies that 
\[
\varphi_2\in[\varphi_1,\varphi_1+\pi]_{2\pi}\cap [\Tf{u}-\pi,\Tf{u}]_{2\pi}.  
  \]
By this together with \eqref{eq-0},  we obtain $\varphi_2\in[\varphi_1,\Tf{u}]_{2\pi}$, which  contradicts \eqref{eq-0} and \eqref{eq-1}, since $\theta(f_u) \not\in (\Tf{v+1},\Tf{u-1})_{2\pi}$. 
Thus we have $\lambda_n \leq \lambda_1+2\pi$, 
completing the proof.
 \end{proof}

By the following lemma, 
we can obtain the proof of Theorem \ref{theorem:main_theorem-x} (ii). 
\begin{lemma}\label{lemma:MainTheorem2}
Let  $\sigma:[n]\to [n]$ be a counterclockwise permutation for monotone linear functions $f_1, \dots , f_n$. 
If it provides the identity, i.e., $\Fall(x)=x$,
then any of the counterclockwise permutations provides the identity.
 \end{lemma}
 The proof follows from the following property of the $k$-shifts of a permutation which produces the identity.  
 
 \begin{lemma}\label{lemma:MainTheorem2-1}
Let  $f_1, \dots, f_n$ be monotone linear functions. 
If a  permutation $\sigma:[n]\to [n]$ provides the identity, i.e., $\Fall(x)=x$, then  any $k$-shift $\sigma_k$ of $\sigma$ provides the identity.
 \end{lemma}
 \begin{proof}
By the equivalence of Lemma \ref{lemma:theta_composition} (iv), $(f_n \circ \dots \circ f_{k+1})\circ (f_{k}\circ \dots \circ f_{1})
=(f_{k} \circ \dots \circ f_{1})\circ (f_{n}\circ \dots \circ f_{k+1})$ for any $k \in \{0, 1, \ldots, n-1\}$.
This means that $\Fall[\sigma_k](x)=\Fall(x)$.
 \end{proof}

\begin{proof}[Proof of Lemma \ref{lemma:MainTheorem2}]
We note that for any permutation $\nu:[n]\to [n]$, 
$\theta(f_{\nu(k)})=\theta(f_{\nu(k+1)})$  implies that $f^{\nu}=f^{\nu_{k,k,k+1}}$ by Lemma \ref{lemma:sin}. 
Since any of the counterclockwise permutations is obtained by repeatedly applying this operation and $k$-shifting to $\sigma$,  Lemma \ref{lemma:MainTheorem2-1} provides the proof. 
\end{proof}

\begin{proof}[Proof of Theorem \ref{theorem:main_theorem-x} (ii)]
(ii-1) $\Longrightarrow$ (ii-2) follows from Lemmas \ref{lemma:localopt_counterclockwise} and \ref{lemma:MainTheorem2}. 

For the converse direction, by Lemma \ref{lemma:localopt_counterclockwise} we suppose, on the  contrary, that all counterclockwise permutations provide the same non-identical function $g$. Since $f_i$'s are not colinear, 
there exists a non-identical linear function $f_i$ such that 
\begin{equation}
\label{eq-lats1}\theta(f_i)\ \not\in \{\theta(g), \theta(g)+\pi \}_{2\pi}.
\end{equation}
Consider a counterclockwise permutation $\sigma:[n]\to [n]$ with $\sigma(1)=i$, and 
let $h=f_{\sigma(n)} \circ \dots \circ f_{\sigma(2)}$. 
Then we have $g=h \circ f_i$. 
Since $\theta(h)\not\in\{\theta(f_i), \theta(f_i)+\pi\}_{2\pi}\cup\{\bot\}$ by \eqref{eq-lats1} and Lemma \ref{lemma:vector_composition}, 
Lemma \ref{lemma:sin} (i) implies that 
$h \circ f_i \not= f_i \circ h$, which contradicts the assumption. 
\end{proof}
\subsection{Proof of Theorem \ref{theorem:main_theorem-x} (iii)}
We prove 
Theorem \ref{theorem:main_theorem-x} (iii) by showing the unimodality of $f^\sigma$ for counterclockwise permutations $\sigma$. 

For a (finite) cyclically ordered set $E=\{ e_1, \ldots, e_m\}$ with $e_1 \prec e_2 \prec \dots\prec e_m   \prec e_1$ and  a weakly ordered set $S$,   a function $g:E\to S$  is called  {\em unimodal}  if there exists two integers $k$ and $\ell$ in $[m]$ such that $g(e_k) \leq g(e_{k+1})\leq \cdots \leq g(e_\ell) \geq g(e_{\ell+1}) \geq \cdots \geq g(e_{k-1})\geq g(e_k)$.   
For a permutation $\tau:[n]\to [n]$, we regard the set $E_\tau=\{\tau_1, \ldots, \tau_{n-1}, \tau_n\,(=\tau) \}$ of its $k$-shifts as a cyclically ordered set and let $g(\sigma)= f^\sigma$ for $\sigma \in E_\tau$.

We show that $f^\sigma$ is unimodal for counterclockwise permutations $\sigma$ by providing  the following adjacent property of $f^\sigma$, where an illustrative example is presented  
in 
 Example \ref{ex-unimodals}. 

 \begin{lemma}\label{lemma:yamatani}
Let $f_1, \dots , f_n$ be monotone linear functions. For a permutation $\sigma: [n]\to [n]$, we have
\begin{eqnarray}
f^{\sigma} \ < \ \min\{ f^{\sigma_1},f^{\sigma_{n-1}} \} &\Rightarrow &\theta(f^\sigma)+\pi \ \in \ (\Tfp{n},\Tfp{1})_{2\pi},\label{eq-new25}\\
  f^{\sigma}  \ > \ \max\{ f^{\sigma_1},f^{\sigma_{n-1}}\} &\Rightarrow& \theta(f^{\sigma}) \ \in \  (\Tfp{n},\Tfp{1})_{2\pi}.\label{eq-new26}
\end{eqnarray}
 \end{lemma}
 
 \begin{proof}
Without loss of generality, we assume that $\sigma$ is the identity permutation.
To prove 
\eqref{eq-new25},  
we assume that  
$f^{\sigma} < f^{\sigma_1}$ and $f^{\sigma} <  f^{\sigma_{n-1}}$. 
Then we have 
     \begin{eqnarray*}
 f^{\sigma} < f^{\sigma_1}
  &\Leftrightarrow& \Tflr{n}{2} - \Tf{1} \in (0,\pi)_{2\pi}\\
  &\Leftrightarrow& \theta( f^{\sigma}) \in (\Tf{1},\Tf{1}+\pi)_{2\pi}\\
  &\Leftrightarrow &\theta( f^{\sigma})+\pi \in (\Tf{1}-\pi,\Tf{1})_{2\pi},
  \end{eqnarray*}
where the first and second equivalences follow from Lemmas \ref{lemma:sin} (i) and \ref{lemma:theta_composition} (i), respectively.
Similarly, by Lemmas \ref{lemma:sin} (i) and \ref{lemma:theta_composition},  we obtain 
\begin{eqnarray*}
  f^{\sigma} <  f^{\sigma_{n-1}}
  &\Leftrightarrow& \Tf{n} - \Tflr{n-1}{1}\in (0,\pi)_{2\pi}\\
  &\Leftrightarrow &\theta(f^{\sigma}) \in (\Tf{n}-\pi,\Tf{n})_{2\pi}\\
  &\Leftrightarrow &\theta(f^{\sigma})+\pi \in (\Tf{n},\Tf{n}+\pi)_{2\pi}.
     \end{eqnarray*}
Thus we have \[
\theta(f^{\sigma})+\pi \in (\Tf{1}-\pi,\Tf{1})_{2\pi} \, \cap \, (\Tf{n},\Tf{n}+\pi)_{2\pi} \ \subseteq \  (\Tf{n},\Tf{1})_{2\pi}.
\]

To prove \eqref{eq-new26}, assume that 
$f^{\sigma}   >  f^{\sigma_1}$ and $f^\sigma > f^{\sigma_{n-1}}$. Then by applying an argument similar to the proof  of  \eqref{eq-new25}, 
we obtain 
\begin{eqnarray*}
  f^{\sigma} > f^{\sigma_1}
  &\Leftrightarrow& \Tf{1} - \Tflr{n}{2} \in (0,\pi)_{2\pi}\\
  &\Leftrightarrow& \theta(f^{\sigma}) \in (\Tf{1}-\pi,\Tf{1})_{2\pi},\\
  f^{\sigma} > f^{\sigma_{n-1}}
  &\Leftrightarrow& \Tflr{n-1}{1} - \Tf{n} \in (0,\pi)_{2\pi}\\
  &\Leftrightarrow& \theta(f^{\sigma}) \in (\Tf{n},\Tf{n}+\pi)_{2\pi}, 
\end{eqnarray*}
which implies 
\[
\theta(f^{\sigma}) \in (\Tf{1}-\pi,\Tf{1})_{2\pi} \, \cap \, (\Tf{n},\Tf{n}+\pi)_{2\pi} \ \subseteq \  (\Tf{n},\Tf{1})_{2\pi}.
\]
 \end{proof}

 \begin{lemma}\label{lemma:unimodal}
Let $\tau:[n]\to [n]$ be a counterclockwise permutation for monotone linear functions $f_1, \dots , f_n$.   
Then $f^{\sigma}$ is unimodal for permutations $\sigma \in E_\tau$.
\end{lemma}

\begin{proof}
If linear functions are colinear or potentially identical, 
then 
 the statement in the lemma holds by  
Theorem \ref{theorem:main_theorem-x} (i) and Lemmas  \ref{lemma:localopt_counterclockwise} 
 and \ref{lemma:MainTheorem2}.  
Otherwise, by Lemma \ref{lemma:localopt_counterclockwise} we can take $\sigma_k$ that is optimal and rename $\sigma_k$ to $\sigma$. 
We will show that $(f^{\sigma}, f^{\sigma_{1}}, \dots , f^{\sigma_{n-1}})$ is unimodal.

Suppose, on the contrary, that 
the sequence is not unimodal. 
Define two positive integers $\ell$ and $r$ by 
\begin{align*}
  \ell&=\min\{i\in[n-2]\mid f^{\sigma_i}>f^{\sigma_{i+1}}\} ,
  \\
  r&=\min\{i\in[n-2]\mid i > \ell,  f^{\sigma_i}<f^{\sigma_{i+1}}\}.
\end{align*}
Then we have 
\begin{equation}
\label{eq--eq1}
\max\{  f^{\sigma}, f^{\sigma_{r}}\}< \min \{f^{\sigma_\ell}, f^{\sigma_{r+1}}\}.
\end{equation}
Let $g_1,g_2,g_3$, and $g_4$ be monotone linear functions defined by 
     \begin{align*}
  g_1 &= \Flr{\sigma(\ell)}{\sigma(1)},\\
  g_2 &= \Flr{\sigma(r)}{\sigma(\ell+1)},\\
  g_3 &= f_{\sigma(r+1)},\\
  g_4 &= \Flr{\sigma(n)}{\sigma(r+2)}.
     \end{align*}
     Then we can see that 
\begin{align*}
  f^{\sigma} &=g_4\circ g_3\circ g_2\circ g_1, \\
 f^{\sigma_\ell}&= g_1\circ g_4\circ g_3\circ g_2,\\
f^{\sigma_r} &= g_2\circ g_1\circ g_4\circ g_3,\\
 f^{\sigma_{r+1}}&= g_3\circ g_2\circ g_1\circ g_4.
\end{align*}
Since $f^{\sigma}<f^{\sigma_\ell}$, 
$g_1,\dots,g_4$ are not colinear by Theorem \ref{theorem:main_theorem-x} (i).
Hence it follows from Lemma \ref{lemma:localopt_counterclockwise} that  $(g_1,\dots,g_4)$ 
is counterclockwise.
Moreover, by Lemma \ref{lemma:yamatani} and \eqref{eq--eq1}, we have 
\begin{equation}
\label{eq--eq2}
\begin{array}{rll}
  \theta(g_4\circ g_3\circ g_2\circ g_1)+\pi &\in  & (\theta(g_4),\theta(g_1))_{2\pi},\\
  \theta(g_1\circ g_4\circ g_3\circ g_2) & \in  &(\theta(g_1),\theta(g_2))_{2\pi},\\
  \theta(g_2\circ g_1\circ g_4\circ g_3)+\pi & \in &(\theta(g_2),\theta(g_3))_{2\pi},\\
  \theta(g_3\circ g_2\circ g_1\circ g_4) & \in  &(\theta(g_3),\theta(g_4))_{2\pi}.
\end{array}     
\end{equation}
We show that \eqref{eq--eq2} derives a contradiction by separately considering the following three cases. 

\smallskip

{\bf Case 1}.  If $\prod_{i\in[n]} \alpha(f_i)=1$, then  we have  $\theta(f^\nu)\in\{0,\pi\}$ for any permutation $\nu:[n]\to [n]$.
     Therefore, \eqref{eq--eq2}  implies that  $\{0,\pi\} \cap (\theta(g_i),\theta(g_{i+1}))_{2\pi} \neq \emptyset$ for any $i\in[4]$, where we denote $g_5 = g_1$.
However, this contradicts that $(g_1,\dots,g_4)$ is counterclockwise.

\smallskip
{\bf Case 2}. If  $\prod_{i\in[n]} \alpha(f_i)<1$, then we have $\theta(f^\nu)\in(0,\pi)$ for any permutation $\nu:[n]\to [n]$.
Therefore, the following four sets are nonempty:
$(\theta(g_1),\theta(g_2))_{2\pi} \cap (0,\pi)$, 
$(\theta(g_2),\theta(g_3))_{2\pi}\cap (\pi,2\pi)$,  
$(\theta(g_3),\theta(g_4))_{2\pi} \cap (0,\pi)$, and  $(\theta(g_4),\theta(g_1))_{2\pi} \cap (\pi,2\pi)$. 
This again contradicts that $(g_1,\dots,g_4)$ is counterclockwise.

\smallskip
{\bf Case 3}. Otherwise (i.e.,  $\prod_{i\in[n]} \alpha(f_i)>1$), 
it can be proven similarly to Case 2.  
In fact, $\theta(f^\nu)\in(\pi,2\pi)$ holds for any permutation $\nu:[n]\to [n]$.
Hence, we have the following four nonempty sets:
$(\theta(g_1),\theta(g_2))_{2\pi} \cap  (\pi,2\pi)$,
$(\theta(g_2),\theta(g_3))_{2\pi}\cap  (0,\pi)$, 
$(\theta(g_3),\theta(g_4))_{2\pi} \cap  (\pi,2\pi)$, and  
 $(\theta(g_4),\theta(g_1))_{2\pi}\cap  (0,\pi)$, which again contradicts that $(g_1,\dots,g_4)$ is counterclockwise.
 \end{proof}

Example \ref{ex-unimodals} provides an instance for the unimodality.  
In the next subsection, we further prove a stronger property of counterclockwise permutations. 

\begin{example}
\label{ex-unimodals}
Consider the following five monotone linear functions
\[
f_1=\frac{1}{2}x+1, \
f_2=\frac{1}{3}x-1, \
f_3=2x-2 \ 
f_4=2x-1,  \mbox{ and }\
f_5=3x.  
 \]
Then their vectors are given as follows (See Figure \ref{figure:ex--cc}):
 \[
\Ora{f_1}=\Mto{1}{\frac{1}{2}}, \
\Ora{f_2}=\Mto{-1}{\frac{2}{3}}, \
\Ora{f_3}=\Mto{-2}{-1}, \
\Ora{f_4}=\Mto{-1}{-1}, \mbox{ and }\ 
\Ora{f_5}=\Mto{0}{-2}. 
 \]
Note that the identity permutation ${\rm id}:[n]\to [n]$ is counterclockwise for $f_i$'s, and moreover, by Lemma \ref{lemma:localopt_counterclockwise}, we can see that it is optimal, since  
\[
f^{{\rm id}}=2x-23, \,
f^{{\rm id}_1}=2x-\frac{27}{2}, \,
f^{{\rm id}_2}=2x-\frac{19}{6}, \,
f^{{\rm id}_3}=2x-\frac{13}{3}, \, 
f^{{\rm id}_4}=2x-\frac{23}{3}, 
 \]
which also shows that $(f^{{\rm id}}, f^{{\rm id}_1}, f^{{\rm id}_2},f^{{\rm id}_3},f^{{\rm id}_4})$ is unimodal. 
 \begin{figure}[htb]
 \centering
 \includegraphics[width=70mm]{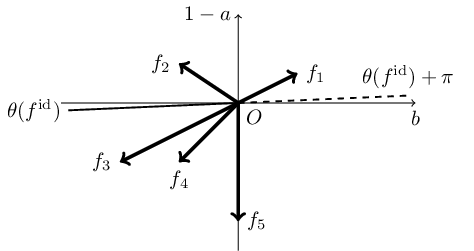}
  \caption{The vector representation for $f_1, \ldots, f_5$.}
  \label{figure:ex--cc}
\end{figure}
 \end{example}

 \begin{proof}[Proof of Theorem \ref{theorem:main_theorem-x} (iii)]
 Without loss of generality, we can assume that all $f_i$'s are non-identical. 
To show the only-if part, let us assume without loss of generality that 
an optimal permutation $\sigma:[n]\to [n]$
for $f_i$'s is the identity. 
Then  $\sigma$ is counterclockwise for $f_i$'s by Lemma \ref{lemma:localopt_counterclockwise}.
We also note that  $\Tflr{n}{2}$ and 
  $\Tflr{n-1}{1}$ are both non-identical by Lemma \ref{lemma-cor:no_x}. 
 Since $\sigma$ is optimal, Lemma \ref{lemma:sin} implies 
     \begin{align*}
  \Tflr{n}{2}-\Tf{1}, \  
  \Tf{n}-\Tflr{n-1}{1} \,\, \in \,\, [0,\pi]_{2\pi}.  
     \end{align*}
By this together with Lemma \ref{lemma:theta_composition}, we have 
     \begin{align*}
  \theta(f^{\sigma})&\ \in \ [\Tf{1},\Tf{1}+\pi]_{2\pi} \ \cap \ [\Tf{n}-\pi,\Tf{n}]_{2\pi},  
     \end{align*}
which  is equivalent to the condition 
     \begin{align*}
  \theta(f^{\sigma})+\pi& \ \in\ [\Tf{1}-\pi,\Tf{1}]_{2\pi} \ \cap \
 [\Tf{n},\Tf{n}+\pi]_{2\pi}.  
\end{align*}
Therefore, we have  $\theta(f^{\sigma})+\pi\in[\Tf{n},\Tf{1}]_{2\pi}$, which completes the proof of the only-if part.     
     
 
For the if-part, we assume that the identity $\sigma\,(={\rm id})$ is  counterclockwise for $f_i$'s and satisfies  that $\theta(f^{\sigma})+\pi\in[\Tf{n},\Tf{1}]_{2\pi}$. 
We first show the following two containment relationships:
     \begin{align}
  \Tf{1}&\ \in \ [\theta(f^{\sigma})-\pi,\theta(f^{\sigma}))_{2\pi}\label{equation:hoge},\\
  \Tf{n}&\ \in \ (\theta(f^{\sigma}),\theta(f^{\sigma})+\pi]_{2\pi}\label{equation:huga}.
     \end{align}
     To prove \eqref{equation:hoge}, suppose, on the contrary, that $\Tf{1}\in[\theta(f^{\sigma}),\theta(f^{\sigma})+\pi)_{2\pi}$.
     Then it follows that 
     \begin{align}
     \label{eqlg1}
\theta(f^{\sigma})+\pi\ \in \ (\Tf{1},\Tf{1}+\pi]_{2\pi} \ \cap \ [\Tf{n},\Tf{1}]_{2\pi}, 
     \end{align}
which implies that    
  $\Tf{n} \ \in \ (\Tf{1},\Tf{1}+\pi]_{2\pi}$. 
Since $\sigma$ is counterclockwise, it holds that $\Tf{i}\in [\Tf{1},\Tf{1}+\pi]_{2\pi}$ for any $i\in[n]$.
By this together with non-colinearity of $f_i$'s,  Lemma \ref{lemma:sin} implies 
that $\theta(f^{\sigma}) \in (\Tf{1},\Tf{n})_{2\pi} \subseteq (\Tf{1},\Tf{1}+\pi)_{2\pi}$, which contradicts \eqref{eqlg1}. 
Thus we have \eqref{equation:hoge}.

Similarly, to prove \eqref{equation:huga}, 
we suppose, on the contrary, that $\Tf{n}\in(\theta(f^{\sigma})-\pi,\theta(f^{\sigma})]_{2\pi}$.
     Then it follows that 
     \begin{align}
     \label{eqlg2}
\theta(f^{\sigma})+\pi\ \in \ [\Tf{n}-\pi,\Tf{n})_{2\pi} \ \cap \ [\Tf{n},\Tf{1}]_{2\pi}, 
     \end{align}
which implies that    
  $\Tf{1} \ \in \ [\Tf{n}-\pi,\Tf{n})_{2\pi}$.
Since $\sigma$ is counterclockwise, it holds that $\Tf{i}\in [\Tf{n}-\pi,\Tf{n}]_{2\pi}$ for any $i\in[n]$.
By this together with non-colinearity of $f_i$'s,  Lemma \ref{lemma:sin} implies 
that $\theta(f^{\sigma}) \in (\Tf{1},\Tf{n})_{2\pi} \subseteq (\Tf{n}-\pi,\Tf{n})_{2\pi}$, which contradicts \eqref{eqlg2}. 
Thus we have \eqref{equation:huga}.

Let us denote two integers $\ell$ and $r$ by 
     \begin{align*}
  \ell=\min\{i\in[n]\mid\Tf{i}\neq_{2\pi}\theta(f^{\sigma})+\pi\}, \\
  r=\max\{i\in[n]\mid\Tf{i}\neq_{2\pi}\theta(f^{\sigma})+\pi\}.
     \end{align*}
Then it follows from  Lemmas \ref{lemma:theta_composition=x1}, \ref{lemma:sin}, and \ref{lemma:theta_composition} that $f^{\sigma_i}=f^{\sigma}$ for any $i$ with $0\leq i < \ell$ or $r\leq i < n$. 
By Lemma \ref{lemma:unimodal}, we only need to show that $f^{\sigma}<f^{\sigma_{\ell}}$ and $f^{\sigma}<f^{\sigma_{r-1}}$.

By Lemma  \ref{lemma:sin}, we note that      
  $f^{\sigma}\,(=f^{\sigma_{\ell-1}})<f^{\sigma_{\ell}}$ 
if and only if  $
    \theta(f_{\ell-1}\circ \dots \circ f_1 \circ f_n \circ \dots \circ f_{\ell+1} )-\theta(f_\ell) \in (0,\pi)_{2\pi} $, 
    which is also equivalent to 
 $\Tf{\ell}\in( \theta(f^{\sigma})-\pi,\theta(f^{\sigma}))_{2\pi}$.
If $\ell=1$, then 
this is satisfied by \eqref{equation:hoge}, 
since $\Tf{1} \not=_{2\pi}\theta(f^{\sigma})-\pi$ by the definition of $l$. 
If $\ell>1$, we again have $\Tf{\ell} \in (\theta(f^{\sigma})-\pi, \theta(f^\sigma))_{2\pi}$ by 
 Lemma \ref{lemma:sin}, 
since otherwise, $\theta(f_i) \in [\theta(f^{\sigma}), \theta(f^\sigma)+\pi]_{2\pi}$ holds for all $i$, which together with their non-colinearity implies that 
$\theta(f^\sigma) \in (\theta(f^{\sigma}), \theta(f^\sigma)+\pi)_{2\pi}$, a contradiction. 
Therefore, we can conclude that 
$f^{\sigma}<f^{\sigma_{\ell}}$. 

Similarly,  
we note that  
 $f^{\sigma}\,(=f^{\sigma_{r}})<f^{\sigma_{r-1}}$ 
if and only if  $
\theta(f_r) -    \theta(f_{r-1}\circ \dots \circ f_1 \circ f_n \circ \dots \circ f_{r+1} )\in (0,\pi)_{2\pi} $,   which is also equivalent to 
 $\Tf{r}\in( \theta(f^{\sigma}),\theta(f^{\sigma})+\pi)_{2\pi}$.
If $r=n$, then 
this is satisfied by \eqref{equation:huga}, 
since $\Tf{n} \not=_{2\pi}\theta(f^{\sigma})+\pi$ by the definition of $r$. 
If $r<n$, 
we again have $\Tf{r} \in (\theta(f^{\sigma}), \theta(f^\sigma)+\pi)_{2\pi}$ by 
 Lemma \ref{lemma:sin}, 
since otherwise 
 $\theta(f_i) \in [\theta(f^{\sigma})-\pi, \theta(f^\sigma)]_{2\pi}$ holds for all $i$, which together with their non-colinearity implies that 
$\theta(f^\sigma) \in (\theta(f^{\sigma})-\pi, \theta(f^\sigma))_{2\pi}$, a contradiction. 
Therefore, we can conclude that 
$f^{\sigma}<f^{\sigma_{r-1}}$, which completes the proof. 
 \end{proof}

\subsection{Structural Properties of Composition of Monotone Linear Functions.} 
We have proved the various properties of the optimal permutations so far. 
Using them, we will show the local structure of composition of monotone linear functions. 
We obtain the following interesting properties. 

 \begin{theorem}\label{theorem:localopt_opt}
    For monotone linear functions $f_1, \dots, f_n$,  a permutation $\sigma:[n]\to [n]$ is optimal if and only if it is locally optimal.
 \end{theorem}
 
\begin{proof}
 Note that the only-if part is immediate from the definitions, 
 and moreover, the if part follows from Theorem \ref{theorem:main_theorem-x} (i),  if $f_i$'s are colinear. 
Thus, 
 we only consider the if part when 
$f_i$'s  are not colinear.   
Let $\sigma:[n]\to [n]$ be a locally optimal permutation for $f_i$'s.  
     By lemma \ref{lemma:localopt_counterclockwise}, the $\sigma$ is counterclockwise.
Hence Theorem \ref{theorem:main_theorem-x} (ii) implies that $\sigma$ is optimal for $f_i$'s if they are potentially identical. 
On the other hand, if they are not potentially identical, then 
by Theorem \ref{theorem:main_theorem-x} (iii), 
for any counterclockwise permutation $\nu$, there exists an optimal permutation $\mu\in N(\nu)$. This implies that $\sigma$ must be optimal, which completes the proof. 
 \end{proof}

For a  cyclically ordered set $E=\{ e_1, \ldots, e_m\}$  and  a weakly ordered set $S$,   a  function $g:E\to S$  is called  {\em strictly unimodal}  if
1) it is unimodal and  
2) $g(e_k)=g(e_{k+1})$ implies that 
they attain either minimum or maximum.

 \begin{theorem}\label{theorem:strong_unimodal}
Let $\tau:[n]\to [n]$
be a counterclockwise permutation for monotone non-identical linear functions $f_1, \dots , f_n$.  
Then $f^{\sigma}$ is strictly unimodal for permutations $\sigma \in E_\tau$.
Furthermore, if $f_i's$ 
 are not potentially identical or colinear, then $f^{\sigma_k} = f^{\sigma_{k+1}}$
 implies  exactly one of the following two conditions.  
     \begin{description}
  \item[{\rm (i)}] $\sigma_k$ and $\sigma_{k+1}$ are minimum permutations in their shifts $($i.e., $f^{\sigma_k} = f^{\sigma_{k+1}}=\min_{0\leq i < n}f^{\sigma_i})$ \ such that    \ $\Tfp{k+1}+\pi=_{2\pi} \theta(f^{\sigma_k})$.
  \item[{\rm (ii)}] $\sigma_k$ and $\sigma_{k+1}$ are maximum permutations in their shifts $($i.e., $f^{\sigma_k} = f^{\sigma_{k+1}}=\max_{0\leq i < n}f^{\sigma_i})$ \ such that    \ $\Tfp{k+1}=_{2\pi} \theta(f^{\sigma_k})$.
      \end{description}
 \end{theorem}
 \begin{proof}
Since the second statement of the theorem implies strict unimodality of $f^\sigma$ for counterclockwise permutations, we only show the second one. 
We assume without loss of generality  that $\sigma$ is the identity  and $k=0$.
By Lemma \ref{lemma:sin},  $f^\sigma=f^{\sigma_1}$ implies that $\theta(f^\sigma)$, $\Tf{1}$, and $\Tflr{n}{2}$ are colinear.
Hence, either $\Tf{1}+\pi =_{2\pi} \theta(f^{\sigma})$ or $\Tf{1} = \theta(f^{\sigma})$ holds. 
In the former case,  $f^\sigma$ is  minimum  by Theorem \ref{theorem:main_theorem-x}, which implies (i) in the theorem. 

 For the latter case, let us assume that  $\Tf{1} = \theta(f^\sigma)$. 
 Let $\ell= \min\{i\in[n]\mid\Tf{i}\neq\Tf{1}\}$ and $r= \max\{i\in[n]\mid\Tf{i}\neq\Tf{1}\}$.
Then 
by 
Lemmas \ref{lemma:theta_composition=x1}, \ref{lemma:sin}, and \ref{lemma:theta_composition},  
it holds that 
 \begin{equation}
 f^{\sigma_{r}}=\cdots=f^{\sigma_{n-1}} =f^\sigma=f^{\sigma_1}=\cdots=f^{\sigma_{\ell-1}}.
 \label{eq--48c1}
 \end{equation}
  We shall show that $f^\sigma> f^{\sigma_\ell}$ and $f^\sigma> f^{\sigma_{r-1}}$, which completes the proof by  Lemma \ref{lemma:unimodal}. 

 In order to show that $f^\sigma> f^{\sigma_\ell}$, we claim 
  that $\Tf{\ell} \in (\theta(f^{\sigma}),\theta(f^\sigma)+\pi)_{2\pi}$.
If it is not the case, i.e., $\Tf{\ell}\in [\theta(f^{\sigma})-\pi,\theta(f^{\sigma})]_{2\pi}$, 
then any $i\in[n]$ satisfies  $\Tf{i}\in[\Tf{\ell},\Tf{1}]_{2\pi} \subseteq [\theta(f^{\sigma})-\pi,\theta(f^{\sigma})]_{2\pi}$, since $\sigma$ is counterclockwise. 
Lemmas \ref{lemma:theta_composition=x1} and  
\ref{lemma:theta_composition}, together with non-colinearity of $f_i$'s imply that $\theta(f^\sigma)\in (\theta(f^{\sigma})-\pi,\theta(f^{\sigma}))_{2\pi}$, which is a contradiction.
Thus our claim holds.
By this together with Lemma \ref{lemma:theta_composition} and \eqref{eq--48c1}, 
we have $\theta(f_{\ell-1}\circ \dots \circ f_1\circ f_n \circ \dots \circ f_{\ell+1})\in (\Tf{\ell}-\pi,\theta(f_\ell))_{2\pi}$.
Therefore,  Lemma \ref{lemma:sin} implies that $f^\sigma=f^{\sigma_{\ell-1}}>f^{\sigma_{\ell}}$.

Similarly, we can prove that $f^\sigma> f^{\sigma_{r-1}}$. 
In fact, we have $\Tf{r} \in (\theta(f^\sigma)-\pi,\theta(f^{\sigma}))_{2\pi}$, which together with Lemma \ref{lemma:theta_composition} and \eqref{eq--48c1} implies  
 $\theta(f_{r-1}\circ \dots \circ f_1\circ f_n \circ \dots \circ f_{r+1})\in (\Tf{r},\theta(f_r)+\pi)_{2\pi}$.
Therefore,  by Lemma \ref{lemma:sin}, we have $f^\sigma=f^{\sigma_{r}}>f^{\sigma_{r-1}}$.
 \end{proof}

By this theorem, 
we have the following corollary. 
\begin{corollary}
Let $f_1, \dots, f_n$ be 
monotone non-identical linear functions. 
If a counterclockwise permutation $\sigma:[n]\to [n]$ 
 for them satisfies $f^{\sigma}\leq f^{\sigma_1}$,  $f^{\sigma}\leq f^{\sigma_{n-1}}$,  and $\Tfp{1}\neq \Tfall$, then it is optimal for them.  
\end{corollary}
\begin{proof}
It follows from Theorem \ref{theorem:main_theorem-x} (iii) and Theorem \ref{theorem:strong_unimodal} 
\end{proof}

\section{Composition of Nondecreasing Linear Functions}\label{section:nondecreasing}
In this section, we deal with the case in which 
aIn this section, we characterize optimal composition orderings for  monotone nondecreasing linear functions $f_1, \dots , f_n$, i.e., $\alpha(f_i)\geq 0$ for all $i \in [n]$.
Since the monotone (increasing) case has been treated in Section \ref{section:monotone}, we assume that 
$\alpha(f_i)=0$ for some $i$.
Note that any permutation $\sigma:[n]\to [n]$ provides a constant function $f^\sigma$ if some $f_i$ is a constant. 
We first show that this case can be transformed to the monotone case.

 For a linear function $f$ and a real number $\epsilon$, we define $f^{(\epsilon)}$ as follows
 \begin{align}\label{eq:epsilon}
     f^{(\epsilon)}(x)=
     \begin{cases}
  f(x) 
  & \text{if } \ \alpha(f) \neq 0,\\
  f(x)+\epsilon x
  & \text{if }\ \alpha(f)=0.
     \end{cases}
 \end{align}
 We denote $f^{(\epsilon)}_{\sigma(n)}\circ\cdots\circ f^{(\epsilon)}_{\sigma(1)} $ by $(f^{(\epsilon)})^{\sigma}$.

\begin{lemma}\label{lemma:epsilon}
For any linear functions $f_1, \dots, f_n$,  there exists a real number $r>0$ such that 
 $(f^{(\epsilon)})^{\sigma}\leq (f^{(\epsilon)})^{\rho}$ 
 implies $\Fall\leq \Fall[\rho]$ 
for any two permutations $\sigma, \rho:[n]\to [n]$
and any real $\epsilon$ with $|\epsilon| < r$.  \end{lemma}
\begin{proof}
Let $\sigma, \rho:[n]\to [n]$ be two permutations.
 If $\Fall \leq \Fall[\rho]$ holds, then the statement in the lemma clearly holds. 
We thus consider the case in which  $\Fall > \Fall[\rho]$. 
Define  $\delta$ and $\Delta$ by 
 \begin{align*}
\delta&= \min\{|\beta(\Fall) - \beta(\Fall[\rho])| \mid \beta(\Fall) \neq \beta(\Fall[\rho]) \}, \\
\Delta&= \max\{|\alpha(f_1)|,\dots,|\alpha(f_n)|,|\beta(f_1)|,\dots,|\beta(f_n)|\}.
 \end{align*}
Since $\Delta=0$ implies that $f^\sigma=f^\rho=0$, 
we assume that $\Delta >0$. 
Note that 
$\beta((f^{(\epsilon)})^{\sigma})$ can be expressed as a polynomial in $\epsilon$ of degree at most $n$ with the absolute value of each coefficient at most  $\Delta^n$.
 In addition, $\beta(\Fall)$ is the constant term of 
 a polynomial for $\beta((f^{(\epsilon)})^{\sigma})$, since $\beta((f^{(0)})^{\sigma})=\beta(\Fall)$.  
We thus denote $\beta((f^{(\epsilon)})^{\sigma})$ by $\epsilon P_\sigma(\epsilon)+\beta(\Fall)$, where  
  $P_\sigma(\epsilon)$ is a polynomial of degree at most 
   $n-1$ with the absolute value of each coefficient at most $\Delta^n$.
 Let $r= \min\{1,\frac{\delta}{2n\Delta^n}\}$ and fix $\epsilon$ to satisfy   $|\epsilon|<r$.
 Then  we have
 \begin{align*}
     (f^{(\epsilon)})^{\sigma}-(f^{(\epsilon)})^{\rho}
     &=\beta((f^{(\epsilon)})^{\sigma})-\beta((f^{(\epsilon)})^{\rho})\\
     &=\beta(\Fall)-\beta(\Fall[\rho])+\epsilon(P_{\sigma}(\epsilon)-P_{\rho}(\epsilon))\\ 
     &\geq\delta-\epsilon(|P_\sigma(\epsilon)|+|P_{\rho}(\epsilon)|)\\
     &\geq\delta-\epsilon(2n\Delta^n) \ > \ 0.
 \end{align*}
    \end{proof}
We  remark that  the lemma does not make use of nondecreasing property on $f_i$'s.    
Lemma \ref{lemma:epsilon} states that 
   the optimality for $f_1^{(\epsilon)},\dots,f_n^{(\epsilon)}$ implies the one for $f_1,\dots,f_n$, if 
$|\epsilon|$ is sufficiently small.
 In other words,  an optimal permutation can be computed 
by making use of an algorithm for the monotone case. 

\begin{lemma}
\label{lemma:epslion1}
Let $f_1, \dots , f_n$ be linear functions. 
For a sufficiently small $\epsilon > 0$, 
any optimal permutation for $f^{(\epsilon)}_i$'s is also optimal for $f_i$'s. 
\end{lemma}

\begin{proof}
It follows from Lemma \ref{lemma:epsilon}. 
\end{proof}
   
   However, we need to be careful about $\epsilon$, which will be mentioned in Section \ref{section:general}.
Furthermore, we can note that the converse of  
 Lemma \ref{lemma:epsilon} does not hold in general, 
 which can be seen in the following example. 
  \begin{example}
  \label{ex-aa21}
Let us consider four monotone nondecreasing functions $f_1,f_2,f_3$, and $f_4$ given by 
\begin{equation*}
 f_1=2x+2, \ f_2=x+2, \ f_3=0x+1,\ f_4=2x-3.
 \end{equation*}
Note that they are non-colinear and not potentially identical. 
The identity permutation is counterclockwise and optimal. 
Since $f_3$ is a constant function, we have 
$f_3=f_3\circ f_2\circ f_1 = f_3\circ f_1\circ f_2$, 
and hence $f_4\circ f_3\circ f_1\circ f_2$ provided by a non-counterclockwise permutation is also optimal.   
\end{example}

We now characterize optimal composition orderings for nondecreasing linear functions. 
Let $f_1, \dots , f_n$ be nondecreasing linear functions, at least one of which is 
a constant. 
Then we first note that any permutation $\sigma:[n]\to [n]$ produces a constant function $f^\sigma$.
Moreover, 
for a permutation $\sigma:[n]\to [n]$, let $q_\sigma$ be the largest integer $q \in [n]$ such that  $\Fp{q}$ is a constant, 
i.e., $q_\sigma=\max\{q \mid \alpha(\Fp{q})=0\}$. 
Then as seen in Example \ref{ex-aa21},  
we have $\Fall=\Fplr{n}{q_\sigma}$, which 
implies that $\Fall=\Fall[\sigma\circ\rho]$ for any permutation $\rho$ such that $\rho(i) = i$ for any $i \in [q_\sigma, n]$.
    In other words, the functions composited before the last constant function can be ordered arbitrarily. 

The following two theorems 
respectively correspond to  Theorem \ref{theorem:main_theorem-x} (i) and (iii) for the monotone case, where we define 
\begin{eqnarray}
\beta_{\min} &=&\min\{\beta(f_i)\mid \alpha(f_i)=0\}, \\
\beta_{\rm OPT}&=&\min\{f^\sigma\mid \sigma \in \Sigma\}.  
\end{eqnarray}
Here $\Sigma$ denotes the set of all permutations of $[n]$.


\begin{theorem}\label{theorem:main_theorem_nondecreasing-1}    
Let $f_1,\dots,f_n$ be monotone nondecreasing linear functions at least one of which is 
a constant. 
 Then 
the following three statements are equivalent. 
 \begin{description}
\item[{\rm (i)}] Any $f_i$ satisfies $\Tf{i}\in[\theta(\beta_{\min})-\pi,\theta(\beta_{\min})]_{2\pi}\cup\{\bot\}$.  
\item[{\rm (ii)}] 
$\beta_{\rm OPT}=\beta_{\min}$.  

\item[{\rm (iii)}] Optimal permutations $\sigma:[n]\to [n]$ for $f_i$'s are those that satisfies that $f_{\sigma(q_{\sigma})}, f_{\sigma(q_{\sigma}+1)}, \ldots, $ $f_{\sigma(n)}$ are colinear. 
\end{description}
\end{theorem}

\begin{theorem}\label{theorem:main_theorem_nondecreasing-2}    
Let $f_1,\dots,f_n$ be monotone nondecreasing linear functions at least one of which is 
a constant. 
If there exists  a linear function $f_i$ such that  $\Tf{i}\in(\theta(\beta_{\min}),\theta(\beta_{\min})+\pi)$ $($equivalently, 
$\beta_{\rm OPT} < \beta_{\min})$, 
then the following two statements are equivalent. 
\begin{description}
\item[{\rm (i)}] A permutation $\sigma:[n]\to [n]$ is 
 optimal.
 \item[{\rm (ii)}] 
A permutation $\sigma:[n]\to [n]$  can be composed by  
 $\sigma=\mu\circ\rho$ for   
 a counterclockwise permutation $\mu:[n]\to [n]$ such that  
 $\Tfall[\mu]+\pi\in[\Tfp[\mu]{t},\Tfp[\mu]{s}]_{2\pi}$, where $s$ and $t$ denote the first and last integers $i$ such that $f_{\sigma(i)}$ is non-identical, and a permutation $\rho:[n]\to [n]$ such that $\rho(i) = i$ for any $i \in [q_\sigma, n]$. 
 \end{description}
\end{theorem}

We remark that 
there exists no case corresponding to   Theorem \ref{theorem:main_theorem-x} (ii),  since $\Fall$ is a constant for any permutation $\sigma$.
In this section, 
we also show that the local optimality implies the global optimality for the nondecreasing case.

\begin{theorem}\label{theorem:localopt_opt_nondecreasing}
     For monotone nondecreasing linear functions $f_1, \dots, f_n$, a permutation $\sigma:[n]\to [n]$ is optimal if and only if it is locally optimal.
 \end{theorem}
    

In order to prove   Theorems \ref{theorem:main_theorem_nondecreasing-1}, \ref{theorem:main_theorem_nondecreasing-2} and \ref{theorem:localopt_opt_nondecreasing}, 
we provide a few lemmas. 
 Recall that for a real $c$, the symbol $\theta(c)$  denotes the angle of the linear function $c\,(=0x+c)$. 
 \begin{lemma}\label{corollary:constant}
     For any linear function $r$, we have the following three equivalences. 
     \begin{description}
  \item[{\rm (i) }]
      $r(c)<c  \ \Leftrightarrow \ \theta(r)\in(\theta(c),\theta(c)+\pi)_{2\pi}$.
  \item[{\rm (ii)}]   
      $r(c)=c  \ \Leftrightarrow \ \theta(r)\in\{\theta(c),\theta(c)+\pi, \bot\}$. 
  \item[{\rm (iii)}]     
      $r(c)>c \ \Leftrightarrow \ \theta(r)\in (\theta(c)-\pi,\theta(c))_{2\pi}$.
     \end{description}
 \end{lemma}
 \begin{proof}
 It follows from 
Lemmas \ref{lemma:theta_composition=x1} and  \ref{lemma:sin}. 
\end{proof}
 
\begin{lemma}\label{lemma:basic_for_nondecreasing}
 Let $f_1,\dots,f_n$ be  monotone nondecreasing linear functions at least one of which is a constant.
 Let $\sigma:[n]\to [n]$ be a local optimal permutation for $f_i$'s.
 Then we have the following four statements.
    \begin{description}
 \item[{\rm (i)}] $f^\sigma \leq \beta_{\min}$.
 \item[{\rm (ii)}] $f_{\sigma(q_\sigma)}=\beta_{\min}$.
 \item[{\rm (iii)}] $\Tfp{i} \ \in \ [\theta(f^\sigma)-\pi,\theta(\beta_{\min})]_{2\pi} \, \cup \  \{\bot\}$  for any $i\in[q_{\sigma}]$.
 \item[{\rm (iv)}] $\Tfp{i} \ \in \ [\theta(\beta_{\min}),\theta(f^\sigma)+\pi] \ \cup \ \{\bot\}$ for any $i\in[q_\sigma,n]$.
    \end{description}
    \end{lemma}
    \begin{proof}
 Without loss of generality, we assume that $\sigma$ is the identity.
Let  $i$ be an index with $f_i(x)=\beta_{\min}$.
If $i=n$, then we have  $\Fall[\sigma]=\beta_{\min}$. 
On the other hand if $i <n$, then 
 $\Fall[\sigma_i]=\beta_{\min}$ and 
$ \sigma_i$ is a neighbor of $\sigma$.  
These imply {\rm (i)}.
To show the rest of the  properties, 
we note that the following inequalities hold by the local optimality:  
 \begin{align}
&f_n\circ \dots \circ  f_{q_{\sigma} + 1}\circ f_i\circ\Flr{q_{\sigma}}{i+1} \ \circ f_{i-1} \circ \dots \circ f_1 \geq \ f^\sigma, \label{eq:b_qaa1}\\
&f_i(f^\sigma)\ =\ f_i \circ \Flr{n}{i+1}\circ\Flr{i-1}{1} \ \geq \ f^\sigma 
\label{eq:Ball}
 \end{align}
for any $i\in [q_{\sigma}-1]$. 
Since  $f_{q_{\sigma}}$ is a constant, 
\eqref{eq:b_qaa1} implies 
 \begin{align}
&f_i(f_{q_{\sigma}}) \ =\ f_i\circ\Flr{q_{\sigma}}{i+1} \ \geq \ \Flr{q_{\sigma}}{i}\ =\ f_{q_{\sigma}}.  \label{eq:b_q}
 \end{align}
 Let $f_i$ be a constant, i.e., $\alpha(f_i)=0$. Then 
 by the definition of $q_{\sigma}$, 
 we have $i \leq q_{\sigma}$, which together with 
  \eqref{eq:b_q} implies  {\rm (ii)}, 
  since $f_i=f_i(f_{q_{\sigma}})$. 
 Moreover,  by \eqref{eq:b_q}, \eqref{eq:Ball}, and Lemma \ref{corollary:constant},  
 we have 
  \begin{align*}
     \theta(f_i)
     \in ([\theta(\beta_{\min})-\pi,\theta(\beta_{\min})]_{2\pi} \cup \{\bot\}) \cap ([\theta(f^\sigma)-\pi,\theta(f^\sigma)]_{2\pi} \cup \{\bot\})
     \end{align*}
 for any $i\in[q_{\sigma}-1]$. 
 Since (i) implies $\theta(f^\sigma)\geq  \theta(\beta_{\min})$,  it holds that 
 \begin{align*}
     \theta(f_i)
     \in [\theta(f^\sigma)-\pi,\theta(\beta_{\min})]_{2\pi} \cup \{\bot \},
 \end{align*}
which proves {\rm (iii)}.
 
 For any integer $i\in[q_{\sigma},n]$, define a real $v_i$ by  $v_i=\Flr{i}{q_{\sigma}}$.
 Then by the local optimality of $\sigma$, 
 we have
 $v_{i+1}=f_{i+1} \circ\Flr{i}{q_{\sigma}}\leq\Flr{i}{q_{\sigma}}\circ f_{i+1}=v_{i}$. 
That is, the sequence $v_{q_{\sigma}},v_{q_{\sigma}+1},\dots,v_n$ is non-increasing. 
This implies that 
$v_{i}\in[v_n,v_{q_{\sigma}}]=[f^\sigma,\beta_{\min}]$ for any $i\in[q_{\sigma},n]$.
Moreover, since $v_{i+1}=f_{i+1}(v_{i})\leq v_{i}$, Lemma \ref{corollary:constant} implies that 
$\theta(f_{i+1})\in[\theta(v_{i}),\theta(v_{i})+\pi] \cup \{\bot \}\subseteq[\theta(\beta_{\min}),\theta(f^\sigma)+\pi] \cup \{\bot\}$, which 
completes the proof. 
\end{proof}

We are now ready to prove Theorem \ref{theorem:localopt_opt_nondecreasing}. 


\begin{proof}[Proof of Theorem \ref{theorem:localopt_opt_nondecreasing}]
By Theorem \ref{theorem:localopt_opt},  it is sufficient to prove the case where at least one of $f_i$'s is a constant.
We assume 
without loss of generality that the identity $\sigma$ is local optimal for $f_i$'s.
Recall that $f_{q_\sigma}=\beta_{\min}$ by Lemma \ref{lemma:basic_for_nondecreasing} and $f_n\circ\cdots\circ f_{q_\sigma+1}$ is an optimal function for $f_{q_\sigma+1}, \dots, f_n$ by Theorem \ref{theorem:localopt_opt}.
Supposing to the contrary that $\beta_{\rm OPT}<f^\sigma$, we construct 
 a permutation $\tau:[n]\to [n]$ such that 
\begin{enumerate}[(I)]
    \item $f^{\tau}<f^\sigma$,
    \item $f_{\tau(q_{\tau})}=\beta_{\min}$,
    \item $\{\tau(i) \mid i>q_{\sigma}\} = [q_{\tau}+1,n]$, 
\end{enumerate}
which completes the proof, 
since  the existence of such a permutation $\tau$ contradicts $f_{q_\sigma}=\beta_{\min}$ and the optimality of $f_n\circ\cdots\circ f_{q_\sigma+1}$.

Let $\sigma^*$ be an optimal permutation, 
and let  $v_i=\Fplr[\sigma^*]{i}{q_{\sigma^*}}$ for $i \in [q_{\sigma^*}, n]$. 
By recalling the proof of Lemma \ref{lemma:basic_for_nondecreasing}, 
we have $v_n\leq\dots\leq v_{q_{\sigma^*}}=\beta_{\min}$, and moreover
there exists an optimal permutation $\sigma^*$ such that $v_n<\dots<v_{q_{\sigma^*}}=\beta_{\min}$.
Assume that $\sigma^*$ satisfies the property. Then there exists an $\ell \in[q_{\sigma^*}+1,n]$ such that $v_\ell <f^\sigma\leq v_{\ell-1}$.
Let  $\tau$ denote the $\ell$-shift of $\sigma^*$, i.e., $\tau=\sigma^*_\ell$, 
where we define $\sigma^*_n = \sigma^*$.
Note that this $\tau$ satisfies (I) and (II). 
Moreover, since $v_{i-1}\in[f^\sigma,\beta_{\min}]$ and  $\Fp[\tau]{i}(v_{i-1})<v_{i-1}$  for any $i \in [q_\tau+1, n], $
 Lemmas \ref{corollary:constant} (i) and  \ref{lemma:basic_for_nondecreasing} (iv) imply that 
 $\{\tau(i) \mid i>q_{\tau}\} \subseteq [q_{\sigma}+1,n]$.
Again by  Lemmas \ref{corollary:constant} and \ref{lemma:basic_for_nondecreasing} (iv),  any $i\in[q_\sigma+1,n]$ satisfies  at least one of the following two inequalities:
\begin{align}
    f_i(\beta_{\min})&\leq \beta_{\min} \ \mbox{ and }\ 
    f_i(f^{\tau})\leq f^{\tau}, \label{eq-9930}
\end{align}
which respectively imply that 
\begin{equation}
\label{eq-9931}
\begin{array}{l}
\Fplr[\tau]{n}{q_{\tau}+1} \circ \Fp[\tau]{j} \circ \Fplr[\tau]{q_{\tau}}{1} \  \leq \ f^{\tau} \ \mbox{ and} \\
    \Fp[\tau]{j}\circ\Fplr[\tau]{n}{j+1}\circ\Fplr[\tau]{j-1}{1} \ \leq \  f^{\tau}, 
\end{array}
\end{equation}
if an index  $j \in [q_\tau -1]$ 
satisfies $\tau(j)\in[q_\sigma+1,n]$. 
In either case, a new permutation $\tau$ corresponding to the left hand sides of \eqref{eq-9931}  satisfies 
\eqref{eq-9930}. 
Therefore, by repeatedly applying this  modification, we can obtain a permutation $\tau$ that satisfies (I), (II) and (III), completing the proof.
\end{proof}

Lemma \ref{lemma:basic_for_nondecreasing} and Theorem \ref{theorem:localopt_opt_nondecreasing} 
 state that $\theta(\beta_{\min})$ and $\theta(\beta_{\rm OPT})+\pi$ provide the {\it boundaries} of  optimal permutations.
The following lemma proves the equivalence of (i) and (ii) in Theorem \ref{theorem:main_theorem_nondecreasing-1}. 
   
\begin{lemma}\label{lemma_nondecreasing-11}    Let $f_1,\dots,f_n$ be monotone nondecreasing linear functions at least one of which is 
a constant.
 Then $\beta_{\rm OPT} =\beta_{\min}$ if and only if  $\Tf{i}\in[\theta(\beta_{\min})-\pi,\theta(\beta_{\min})]_{2\pi} \cup \{ \bot\}$ holds for any $i\in[n]$.  
\end{lemma}
\begin{proof}
Suppose on the contrary that there exists a linear function $f_j$ such that  $\Tf{j}\in(\theta(\beta_{\min}),$ $\theta(\beta_{\min})+\pi)_{2\pi}$, and    
let $f_i=0x+\beta_{\min}$. 
Then consider 
 a permutation $\sigma:[n]\to [n]$ with $\sigma(n-1)=i$ and $\sigma(n)=j$.  
Since $f^\sigma=f_j(\beta_{\min})$,  
Lemma \ref{corollary:constant} implies  that $f^\sigma< \beta_{\min}$, which concludes that 
$\beta_{\rm OPT}< \beta_{\min}$. 

On the other hand, if all $f_i$'s satisfy 
$\Tf{i}\in[\theta(\beta_{\min})-\pi,\theta(\beta_{\min})]_{2\pi} \cup \{\bot\}$, 
then for any optimal permutation $\sigma:[n]\to [n]$ for $f_i$'s,  we have 
$f^\sigma= f_{\sigma(n)} \circ \dots \circ f_{\sigma(q_\sigma+1)}(\beta_{\min})$ by Lemma  \ref{lemma:basic_for_nondecreasing} (ii). 
It follows from  Lemmas \ref{lemma:theta_composition=x1} and \ref{lemma:theta_composition} that 
\[\theta(f_{\sigma(n)} \circ \dots \circ f_{\sigma(q_\sigma+1)})\in[\theta(\beta_{\min})-\pi,\theta(\beta_{\min})]_{2\pi}
\cup \{\bot\}, 
\]
which together with Lemmas  \ref{corollary:constant}  and \ref{lemma:basic_for_nondecreasing} (i)
implies that $\beta_{\rm OPT} = \beta_{\min}$. 
\end{proof}

\begin{proof}[Proof of Theorem \ref{theorem:main_theorem_nondecreasing-1}]
By Lemma \ref{lemma_nondecreasing-11}, (i) and (ii) are equivalent. 
If $\beta_{{\rm OPT}}=\beta_{\min}$, 
then Lemma \ref{lemma:basic_for_nondecreasing} (iii) and (iv) imply that 
any optimal permutation $\sigma:[n]\to [n]$ for $f_i$'s  satisfies 
$i<q_\sigma$ for any $i$ with 
 $\theta(f_{\sigma(i)})\in (\theta(\beta_{\min})-\pi,\theta(\beta_{\min}))_{2\pi}$. 
On the other hand, if a permutation $\sigma$ satisfies the condition above, then we have 
$f^\sigma= f_{\sigma(n)} \circ \dots \circ f_{\sigma(q_\sigma+1)}(\beta_{\min})=\beta_{\min}$, 
where the last equality follows from 
Lemma \ref{corollary:constant} (ii). 
This implies that  $\sigma$ is optimal for $f_i$'s, proving the implication of  (ii) $\Longrightarrow$ (iii).  

To show (iii) $\Longrightarrow$ (i), 
let $f_j$ satisfy 
$\Tf{j}\in(\theta(\beta_{\min}),$ $\theta(\beta_{\min})+\pi)$ for some $j$ and assume that the identity {\rm id} is optimal for $f_i$'s.   Then by Lemma \ref{lemma:basic_for_nondecreasing} (iv), 
the inequality $j > q_{\rm id}$ must hold, 
since $\theta(\beta_{\min})\leq  \theta(\beta_{\rm OPT})$. 
However, (iii) does not imply this property, which completes the proof.  
\end{proof}

\begin{proof}[Proof of Theorem \ref{theorem:main_theorem_nondecreasing-2}]
By Lemmas \ref{lemma:basic_for_nondecreasing} (i) and   
\ref{lemma_nondecreasing-11}, 
we have $\beta_{\rm OPT} < \beta_{\min}$ if and only if some $f_j$ satisfies $\Tf{j}\in(\theta(\beta_{\min}),\theta(\beta_{\min})+\pi)$. 

To show (i) $\Longrightarrow$ (ii),  
let us first assume that the identity ${\rm id}:[n]\to [n]$ is optimal for $f_i$'s. 
By Lemma \ref{lemma:basic_for_nondecreasing}, $\theta(\beta_{\min})$ and $\theta(\beta_{\rm OPT})+\pi$ are the boundaries of $\{\theta(f_i)\mid i<q_{\rm id}\}$ and $\{\theta(f_i)\mid i>q_{\rm id}\}$.
Thus we only have to prove that $f_{q_{\rm id}+1},\cdots,f_n$ is counterclockwise.
Since they are monotone, Theorem \ref{theorem:main_theorem-x} implies that it is true if the following condition is not satisfied. 
\begin{equation}
 \{\psi,\psi+\pi\} \ \subseteq  \  
 \{\theta(f_i) \mid i > q_{\rm id}
 \} \ \subseteq \ 
 \{\psi,\psi+\pi,\bot\} \label{eq:psi_00}
\end{equation}
for some $\psi$ with $0 \leq \psi < \pi$.
Note that $\psi=\theta(f_j)$ for some $j$, 
since $\theta(\beta_{\min}) <\theta(\beta_{\rm OPT})$.
Suppose on the contrary that \eqref{eq:psi_00} is satisfied. Then 
by repeatedly applying Lemma \ref{lemma:theta_composition} {\rm (ii)}, we have 
\begin{equation*}
    \theta(\Flr{n}{q_{\rm id}+1})\in\{\psi,\psi+\pi,\bot\}.
\end{equation*}
Since $\beta_{\rm OPT} < \beta_{\min}$,  
Lemma \ref{corollary:constant} (i) implies that 
$\theta(\Flr{n}{q_{\rm id}+1})=\psi$. 
By Lemma \ref{lemma:vector_composition}, 
we have 
\begin{equation*}
  \theta(\beta_{\rm OPT})  \in (\theta(\beta_{\min}),\psi), 
\end{equation*}
which contradicts that 
$\psi+\pi \leq \theta(\beta_{\rm OPT}) + \pi$
by Lemma \ref{lemma:basic_for_nondecreasing} (iv).  

 To show that (ii) $\Longrightarrow$ (i), 
 let $\mu:[n]\to [n]$ be a counterclockwise permutation that satisfies $\theta(f^\mu)+\pi\in[\theta(f_{\mu(n)}),\theta(f_{\mu(1)})]_{2\pi}$.
 Since $f^\mu=f^{\mu\circ \rho}$ for any permutation $\rho:[n]\to [n]$ such that $\rho(i) = i$ for any $i \in [q_\sigma, n]$, it is enough to show that $\mu$ is optimal for $f_i$'s. 
 By (i) $\Longrightarrow$ (ii), which has already been proven, we can take an optimal permutation $\nu$ such that $\theta(f^\nu)+\pi\in[\theta(f_{\nu(n)}),\theta(f_{\nu(1)}]_{2\pi}$ and $\mu_k=\nu$ for some $k$.
 Suppose on the contrary that $f^\nu< f^\mu$. Then we have $\{\theta(f_{\mu(1)}),\dots,\theta(f_{\mu(k)})\}\subseteq[\theta(f^\mu)+\pi,\theta(f^\nu)+\pi]\cup \{\bot\}$ and 
\begin{eqnarray*}
f^\nu&=&f_{\mu(k)} \circ \dots \circ f_{\mu(1)}\circ f_{\mu(n)} \circ \dots \circ  f_{\mu(k+1)} \nonumber\\
&=& f_{\mu(k)} \circ \dots \circ f_{\mu(1)} (f^\mu) 
 \ < \  f^\mu, \nonumber
\end{eqnarray*}
where the second equality follows because any constant function in $f_i$'s is in $\{f_{\mu(k+1)},\dots,f_{\mu(n)}\}$.
This together with Lemma \ref{corollary:constant} implies that 
\begin{equation}
\label{eq-non-las1}
\theta(f_{\mu(k)} \circ \dots \circ f_{\mu(1)} )\in(\theta(f^\mu),\theta(f^\mu)+\pi), 
\end{equation}
which contradicts $\{\theta(f_{\mu(1)}),\dots,\theta(f_{\mu(k)})\}\subseteq[\theta(f^\mu)+\pi,\theta(f^\nu)+\pi] \cup \{\bot\}$.
\end{proof}

We note that the equivalence of (i) and (ii) in Theorem 
\ref{theorem:main_theorem_nondecreasing-2} does not imply     
that $\beta_{\rm OPT}< \beta_{\min}$. 

Before concluding this section, 
we mention the maximization problem. 
\begin{remark}
\label{remark-11}
As discussed in Section \ref{section:notation}, 
the maximization for $f_i$'s is equivalent to the minimization for $\tilde{f}_i$'s given by \eqref{eq--0001}. 
Thus all the results for monotone nondecreasing functions are applicable for the maximization problem. 
Since the mapping \eqref{eq--0001} is the reflection across the $(1-a)$-axis in the vector representation, we can obtain the results by exchanging the term ``counterclockwise'' by ``clockwise''. 
\end{remark}

\section{Composition of General Linear Functions}\label{section:general}
In this section, we discuss the composition for general linear functions $f_1, \dots,$ $f_n$, 
where an example  of composition for general linear functions 
is given in Example \ref{ex-general1}. 
    Let $k$ denote the number of monotone decreasing functions in them, i.e., 
$k=|\{i\in[n]\mid \alpha(f_i)<0\}|$.
In Sections \ref{section:monotone} and \ref{section:nondecreasing}, 
we provided structure characterizations for the composite when $k=0$, and as their corollary we presented a polynomial-time algorithm for the minimization/maximization for monotone nondecreasing linear functions. 
In this section, we present several structural properties for the composites of general linear functions and show fixed-parameter tractability for the minimization problem with respect to $k$, whose complexity status was open \cite{KMS:linear}. 

In the rest of this section,  we mainly restrict our attention to the case in which no linear function
is identical or constant, i.e., $f_i\not=x$ and $\alpha(f_i)\not= 0$ for all $i \in [n]$.
Note that the identical function plays no role in optimal composition, and 
Lemma \ref{lemma:epslion1} in Section \ref{section:nondecreasing} implies that  
the general composition problem  can be transformed to this case, by using $f_i^{(\epsilon)}$'s  for some $\epsilon>0$. 
We remark that our algorithm does not make use of $\epsilon$ explicitly, since the orderings of angles $\theta(f_i^{(\epsilon)})$'s are only needed. 

\begin{example}   
\label{ex-general1}
Let $f_i$ $(i=1, \dots 7)$ be linear functions defined by 
    \begin{align*}
 &f_1=\frac{1}{3}x, \ f_2=\frac{2}{3}x+1,\ f_3=x+\frac{1}{2},\ f_4=-x-3,\ f_5=x-1, \ f_6=\frac{3}{2}x, \ f_7=2x+1, 
    \end{align*}
where all but $f_4$ are monotone. The vector representation is shown in Fig. \ref{figure:vectors}
 \begin{figure}[htb]
     \centering
     \includegraphics[width=60mm]{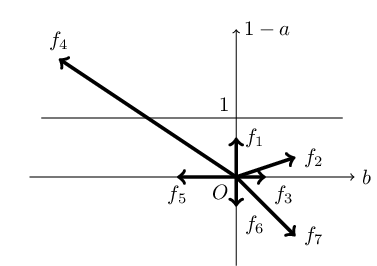}
     \caption{The vector representation for $f_1, ..., f_7$.}
     \label{figure:vectors}
 \end{figure}
Note that the identity permutation provides an optimal function $\Flr{7}{1}$. 
Recall  that
$\Tfp{i+1}-\Tfp{i} \in [0,\pi]_{2\pi}$ holds for 
any optimal (i.e., minimum) permutation $\sigma$ of monotone linear functions by Lemma \ref{lemma:sin}. 
However, this crucial property for monotone linear functions
does not hold in general. 
For example, the instance satisfies $\theta(f_2)-\theta(f_1)\in  (\pi, 2\pi)$. 
Instead, we point out 
the following properties:
$f_3\circ f_2\circ f_1$ is provided by a  maximum permutation for $f_1,f_2$, and $f_3$, while $f_7\circ f_6\circ f_5$ is provided by a minimum permutation for  $f_5,f_6$, and $f_7$.
\end{example}


To see the properties mentioned in Example \ref{ex-general1}, 
We define two sets $L^\sigma$ and $U^\sigma$ of monotone nondecreasing linear functions.  
 For a permutation $\sigma:[n]\to [n]$, 
 let $n^\sigma_1, \ldots, n^\sigma_k$ be integers such that 
 \[
 0 \ (= n^\sigma_0) \ < \ n^\sigma_1 \ < \ \cdots\  < \ n^\sigma_k\ < \   n+1 \ (=n^\sigma_{k+1})
 \] and 
 $\alpha(f_{\sigma(n^\sigma_j)})<0$ for all $j \in [k]$. 
 For $j \in \{0, 1,\dots , k\}$, let  
\begin{align*}
  I_j^{\sigma} &=\{ i\in[n] \mid n^{\sigma}_j< i < n^{\sigma}_{j+1} \}
\end{align*}
and define
\[
L^\sigma=\bigcup_{k-j: \text{even}} I_j^\sigma \ \text{ and } \ U^\sigma=\bigcup_{k-j: \text{odd}} I_j^\sigma. 
\]

By definition, 
the set of indices of all monotone nondecreasing functions $\{i\in[n]\mid \alpha(f_{\sigma(i)})\geq 0\}$ is partitioned into   $L^{\sigma}$ and $U^{\sigma}$. 
In the example in Example \ref{ex-general1}, we have $L^{{\rm id}}=I_1^{{\rm id}}=\{5,6,7\}$ and $U^{{\rm id}}=I_0^{{\rm id}}=\{1,2,3\}$. 

 For an integer $j \in \{0, \dots , k\}$, define linear functions 
\[
r_{i-n^\sigma_j}=f_{\sigma(i)} \ \
\text{ for }\ i \in I_j^\sigma. 
\]

\begin{lemma}\label{lemma:LU_clockwise1}
Let $\sigma:[n]\to [n]$ be an optimal permutation for non-constant and non-identical linear functions $f_1, \dots , f_n$. 
For an integer $j \in \{0, \dots , k\}$, 
let $r_i$ $(i \in I_j^\sigma)$ be defined as above. 
Then we have the following two statements. 
\begin{description}
\item[{\rm (i)}] 
If $k-j$ is even, then the identity ${\rm id}:[ |I_j^\sigma|] \to [ |I_j^\sigma|] $ is a minimum permutation for $r_1, \dots , r_{|I_j^\sigma|}$.  

\item[{\rm (ii)}]
If $k-j$ is odd, then the identity ${\rm id}:[ |I_j^\sigma|] \to [ |I_j^\sigma|] $ is a  maximum permutation for $r_1, \dots , r_{|I_j^\sigma|}$.  
\end{description}
\end{lemma}

\begin{proof}
 Define $p=f_{\sigma(n)} \circ \dots \circ f_{\sigma(n^\sigma_{j+1})}$
and 
$q=f_{\sigma(n^\sigma_{j})} \circ \dots \circ f_{\sigma(1)}$. 
Then  we have $f^\sigma=p\circ r_{|I_j^\sigma|} \circ \dots \circ r_1 \circ q$. 
If $k-j$ is even, then 
$p$ is monotone (increasing). 
This implies that 
$r_{|I_j^\sigma|} \circ \dots \circ r_1$ is provided by a minimum permutation for $r_i$'s. 
On the other hand, if $k-j$ is odd, 
then $p$ is monotone decreasing. 
We can see that $r_{|I_j^\sigma|} \circ \dots \circ r_1$ is provided by a maximum permutation for $r_i$'s. 
\end{proof}

Lemma \ref{lemma:LU_clockwise1}, together with Theorem \ref{theorem:main_theorem-x} and Remark \ref{remark-11},
implies that each interval $I_j$ is permuted either counterclockwisely or  clockwisely, unless all the functions in $I_j$ are colinear. The following lemma states that 
not only intervals $I_j$'s but also $L^\sigma$ and $U^\sigma$ satisfy such a property if $\sigma$ is optimal. 
Let 
$L^{\sigma} = \{ \ell_1,  \dots , \ell_{|L^{\sigma}|}\} $ and   
$U^{\sigma}=\{u_1, \dots , u_{|U^\sigma|}\}$, 
where $\ell_1 < \dots < \ell_{|L^\sigma|}$ and 
$u_1 < \dots < u_{|U^\sigma|}$, and 
let 
\begin{equation*}
p_i=f_{\sigma(\ell_i)} \ \text{ for } \ i \in [|L^{\sigma}|] \ \text{ and } \
q_i=f_{\sigma(u_i)}  \ \text{ for } \  i \in [|U^{\sigma}|]. 
\end{equation*}




\begin{lemma}\label{lemma:LU_clockwise}
Let $\sigma:[n]\to [n]$ be an optimal permutation for non-constant and non-identical linear functions $f_1, \dots , f_n$. 
Let $p_i$ $(i \in [|L^{\sigma}|])$ and $q_i$ $(i \in [|U^{\sigma}|])$ denote monotone linear functions defined as above. 
Then we have the following two statements. 
\begin{description}
\item[{\rm (i)}] 
The identity ${\rm id}:[ |L^{\sigma}|] \to 
 [ |L^{\sigma}|] $ is counterclockwise for $p_i$'s, unless they are colinear. 
 

\item[{\rm (ii)}]
The identity ${\rm id}:[ |U^{\sigma}|] \to [ |U^{\sigma}|] $ is clockwise for $q_i$'s, unless  they are colinear.   
\end{description}
    \end{lemma}

\begin{proof}
We first prove the case where $k$ is even. 
Note that 
\begin{align}
f^\sigma & =
\overbrace{p_{|L^{\sigma}|} \circ \dots \circ p_{|L^{\sigma}|-|I_k^\sigma|+1}}^{I^\sigma_k}  \circ \ g_{k/2} \circ \cdots  
\circ  g_{2}  \circ \overbrace{p_{|I_0^{\sigma}|+|I_2^{\sigma}|} \circ \dots \circ p_{|I_0^\sigma|+1}}^{I^\sigma_2}   \circ \ g_{1} \circ   \overbrace{p_{|I^\sigma_0|} \circ \dots \circ p_1}^{I^\sigma_0}
\label{eq-last31}\\
 & =
h_{k/2+1} \,\, \circ \overbrace{q_{|U^{\sigma}|} \circ \dots \circ q_{|U^{\sigma}|-|I_{k-1}^\sigma|+1}}^{I^\sigma_{k-1}}  \ \circ \ h_{k/2} \, \circ \ \cdots \nonumber \\
& \hspace*{.5cm} 
\circ \  h_{3} \, \circ \  \overbrace{q_{|I_1^{\sigma}|+|I_3^{\sigma}|} \circ \dots \circ q_{|I_1^\sigma|+1}}^{I^\sigma_3}  \circ \  h_{2} 
\, \circ \  \overbrace{q_{|I^\sigma_1|} \circ \dots \circ q_1}^{I^\sigma_1} 
\ \circ \ h_{1},\label{eq-last32}
\end{align}
where
\begin{align*}
g_j&=f_{\sigma(n^\sigma_{2j})} \circ  f_{\sigma(n^\sigma_{2j}-1)}\circ \dots  \circ f_{\sigma(n^\sigma_{2j-1})}&&\text{for }j\in \left\{1, \dots , \left\lceil \frac{k}{2} \right\rceil \right\},\\
h_j&=f_{\sigma(n^\sigma_{2j-1})} \circ  f_{\sigma(n^\sigma_{2j-1}-1)}\circ \dots  \circ f_{\sigma(n^\sigma_{2j-2})}&&\text{for }j \in \left\{1, \dots , \left\lceil \frac{k+1}{2} \right\rceil\right\},
\end{align*}
and we set $f_{\sigma(n+1)}=f_{\sigma(0)}=x$. 

Since all the linear functions in the right-hand side of  \eqref{eq-last31} are monotone, 
Theorem \ref{theorem:main_theorem-x} implies (i) of the lemma. 
Since all the linear functions in the right-hand side of  \eqref{eq-last32} but $h_1$ 
 and $h_{k/2+1}$ are monotone, 
Theorem \ref{theorem:main_theorem-x} and Remark \ref{remark-11} imply (ii) of the lemma. 

We next prove the case where $k$ is odd. 
In this case,  we have 
\begin{eqnarray}
f^\sigma & =&
g_{(k+1)/2} \,\, \circ \overbrace{q_{|U^{\sigma}|} \circ \dots \circ q_{|U^{\sigma}|-|I_{k-1}^\sigma|+1}}^{I^\sigma_{k-1}}  \ \circ \  g_{(k-1)/2} \,\, \circ \ \cdots  \nonumber \\[.2cm]
&& \hspace*{.3cm} \circ \  g_{2} \,\, \circ \  \overbrace{q_{|I_0^{\sigma}|+|I_2^{\sigma}|} \circ \dots \circ q_{|I_0^\sigma|+1}}^{I^\sigma_2}  \ \circ \ g_{1} \,\, \circ \  \overbrace{q_{|I^\sigma_0|} \circ \dots \circ q_1}^{I^\sigma_0}.
\label{eq-last33}\\[.2cm]
 & =&
 \overbrace{p_{|L^{\sigma}|} \circ \dots \circ p_{|L^{\sigma}|-|I_k^\sigma|+1}}^{I^\sigma_{k}} \ \circ \  h_{(k+1)/2} \,\,  \circ \ \cdots \nonumber \\[.2cm]
&& \hspace*{.3cm} \circ \  h_{3} \, \circ \  \overbrace{p_{|I_1^{\sigma}|+|I_3^{\sigma}|} \circ \dots \circ p_{|I_1^\sigma|+1}}^{I^\sigma_3}  \circ \  h_{2} 
\, \circ \  \overbrace{p_{|I^\sigma_1|} \circ \dots \circ p_1}^{I^\sigma_1} 
\ \circ \ h_{1}.\label{eq-last34}
\end{eqnarray}
Note that  all the functions in \eqref{eq-last33} and \eqref{eq-last34} but $h_1$ and $g_{(k+1)/2}$ are monotone, and thus (i) and (ii) in the lemma can be proved. 
\end{proof}

By Lemma \ref{lemma:LU_clockwise}, $L^\sigma$ and $U^\sigma$ are permuted counterclockwisely and clockwisely, respectively.  
Moreover, the following crucial lemma shows that they are partitioned by two angles $\psi_1$ and $\psi_2$ (See Lemma \ref{lemma:strong_L_left_R_right} 
 (iii)). 
For an set $I \subseteq [n]$, let  $\theta(I)=\{\theta(f_{\sigma(i)}) \mid i \in I\}$. 
 
\begin{lemma}\label{lemma:strong_L_left_R_right}
There exists an optimal permutation $\sigma:[n]\to [n]$ for  non-constant and non-identical linear functions $f_1, \dots , f_n$ such that 
\begin{description}
\item[{\rm (i)}]
$f_{\sigma(\ell)}$ $(\ell \in L^\sigma)$ are permuted counterclockwisely,
\item[{\rm (ii)}]
$f_{\sigma(u)}$ $(u \in U^\sigma)$ are permuted clockwisely, 
\item[{\rm (iii)}]
 $\theta(L^\sigma)  \subseteq[\psi_1,\psi_2]$ and $\theta(U^\sigma)\subseteq  (\psi_2,\psi_1)_{2\pi}$ for some 
two angles $\psi_1\in(0,\pi)$ and $\psi_2\in(\pi,2\pi)$,   
\item[{\rm (iv)}]
$\theta(I^\sigma_s) \cap \theta(I^\sigma_t) =\emptyset$ for any distinct $s$ and $t$.  
\end{description}
\end{lemma}

In order to prove Lemma \ref{lemma:strong_L_left_R_right}, we first show the following weak partitionability. 
\begin{lemma}\label{lemma:L_left_R_right}
Let $f_1, \dots , f_n$ be non-constant and non-identical linear functions, at least one of which is monotone decreasing. 
 For any optimal solution $\sigma$ for $f_i$'s, there exist two angles $\psi_1\in(0,\pi)$ and $\psi_2\in(\pi,2\pi)$ such that 
    $\theta(L^\sigma)  \subseteq [\psi_1,\psi_2]$ and
   $\theta(U^\sigma) \subseteq  [\psi_2,\psi_1]_{2\pi}$.
\end{lemma}

\begin{proof}
 Without loss of generality, we assume that  the identity {\rm id} is optimal for $f_i$'s. 
 First, we show that if $\theta(f_i)=\pi$ then $i\in L^{\rm id}$ and if $\theta(f_i)=0$ then $i\in U^{\rm id}$.
 Suppose on the contrary that there exists an $i\in U^{\rm id}$ such that $\theta(f_i)=\pi$.
 Let $g=f_n\circ\cdots\circ f_{i+1}$ then $\alpha(g)<0$ by the definition of $U$.
 Then we have $g\circ f_i > f_i\circ g$ by Lemma \ref{lemma:sin} 
 and $\theta(g)\in(0,\pi)$, which contradicts the optimality of {\rm id}.
 Similarly, suppose on the contrary that there exists an $i\in L^{\rm id}$ such that $\theta(f_i)=0$.
 To obtain a contradiction, we divide the discussion into two cases, whether $\alpha(f_j)<0$ for some $j$ with $j>i$ or not.
 If such $j$ exists, then let $j$ be the minimum integer which satisfies the statement above and $g=f_j\circ\cdots\circ f_{i+1}$.
 Then we have $g\circ f_i < f_i\circ g$ by Lemma \ref{lemma:theta_composition}, which contradicts the optimality of {\rm id}, because $\alpha(f_n\circ\cdots\circ f_{j+1})<0$ by the definition of $L^{\rm id}$.
 On the other hand, if no such $j$ satisfies $\alpha(f_j)<0$, then there exists $j<i$ such that $\alpha(f_j)<0$ by the assumption.
 Let $j$ be the maximum integer which satisfies the statement above and $g=f_{i-1}\circ\cdots\circ f_{j}$.
 Then we have $f_i\circ g > g\circ f_i$ by Lemma \ref{lemma:theta_composition}, which contradicts the optimality of {\rm id}, because $\alpha(f_n\circ\cdots\circ f_i)>0$ by the definition of $L^{\rm id}$.

Next, we discuss a monotone $f_i$ such that $\theta(f_i)\not\in\{0,\pi\}$. 
For any $\ell \in L^{\rm id}$ and $u \in U^{\rm id}$, we claim that the following two implications are satisfied.  
 \begin{equation}
 \label{eq-aac0}
 \begin{array}{lll}
0 \ < \ \theta(f_\ell),  \theta(f_u) \ <\ \pi &\Rightarrow &\Tf{\ell} \ \geq  \ \Tf{u},\\
\pi \ < \ \theta(f_\ell),  \theta(f_u) \ <\  2\pi &\Rightarrow& \Tf{\ell} \ \leq \ \Tf{u}.
 \end{array}
 \end{equation}
 We first consider the case $\ell<u$.
 Let $\psi=\Tflr{u-1}{\ell+1}$.
 By definition of $L^{\rm id}$ and $U^{\rm id}$, it holds that $\alpha(\Flr{u-1}{\ell+1})<0$, which implies  $\psi\in(0,\pi)$.
 Since $f_n \circ \dots \circ f_{u+1}$ is monotone decreasing (i.e., $\alpha(f_n \circ \dots \circ f_{u+1})<0$), 
 the optimality of ${\rm id}$ implies that 
 \[
     \psi- \Tf{u}, \
     \Tf{\ell}- \psi \,\,\in \,\, [0,\pi]_{2\pi}. 
 \]
 Hence we have 
 \begin{align}
     &0 \ < \ \Tf{u} \ \leq \ \psi \ \leq \ \Tf{\ell} \ <  \ \pi &&\text{if } \ 0 \ < \ \theta(f_\ell),  \theta(f_u) \ < \ \pi \label{eq-aac1}\\
     &\pi \ < \ \Tf{\ell} \ \leq \ \psi+\pi \ \leq \ \Tf{u} \ < \  2\pi &&\text{if } \ \pi \ < \ \theta(f_\ell),  \theta(f_u) \ <\ 2\pi.\label{eq-aac2}
  \end{align}
 
On the other hand, if $u<\ell$,  let $\psi =\theta(f_{\ell-1} \circ \dots \circ f_{u+1})$.
 Then $\psi \in(0,\pi)$ by the definition of $L^{\rm id}$ and $U^{\rm id}$.
 Since $f_n \circ \dots \circ f_{\ell+1}$ is monotone and ${\rm id}$ is optimal, we have 
 \begin{eqnarray*}
     \psi - \Tf{u}, \,\,
     \Tf{\ell}-\psi \,\,\in \,\, [0,\pi]_{2\pi},
 \end{eqnarray*}
which implies  \eqref{eq-aac1} and \eqref{eq-aac2}. 

Therefore, \eqref{eq-aac0} holds for any $\ell \in L^{\rm id}$ and $u \in U^{\rm id}$, completing the proof. 
\end{proof}

Note that 
 Lemmas
\ref{lemma:LU_clockwise} and \ref{lemma:L_left_R_right} do {\em not} imply the polynomial solvability of optimal composition if $k$, the number of decreasing functions, is a constant,  since it is not clear how to partition linear functions with the same angle. The following lemma states that  some optimal permutation $\sigma$ makes linear functions with the same angle contained in the same interval $I^\sigma_j$.

\begin{proof}[Proof of Lemma \ref{lemma:strong_L_left_R_right}]

Let $\sigma: [n]\to [n]$ be an optimal permutation for $f_i$'s. 
Then it satisfies the conditions  in  Lemmas \ref{lemma:LU_clockwise} and \ref{lemma:L_left_R_right}. 
For two integers $s$ and $t$ with $s<t$, 
assume that  $i_s \in I^\sigma_s$ and $i_t \in I^\sigma_t$ satisfy $\theta(f_{\sigma(i_s)})=\theta(f_{\sigma(i_t)})$. 
Let $g=f_{\sigma(i_t-1)} \circ \dots \circ f_{\sigma(i_s+1)}$, and consider 
the composition $f_{\sigma(i_t)}\circ g \circ f_{\sigma(i_s)}$. 
By Lemma \ref{lemma:sin} and the optimality of $\sigma$, we must have $\theta(g) \in 
\{\theta(f_{\sigma(i_s)}), \theta(f_{\sigma(i_s)})+\pi\}_{2\pi}\cup\{\bot\}$, 
which implies 
$f_{\sigma(i_t)}\circ g \circ f_{\sigma(i_s)} \ = \ g\circ f_{\sigma(i_t)} \circ f_{\sigma(i_s)}$.
Thus the permutation corresponding to this modification is also optimal for $f_i$'s. 
By repeatedly applying this modification, we can arrive at an optimal permutation that satisfies (i), (ii), (iii), and (iv) in the lemma, if no $j$ satisfies 
$\theta(I_j^\sigma)=\{\lambda, \lambda+\pi\}$ for some $\lambda$.
If $\theta(I_j^\sigma)=\{\lambda, \lambda+\pi\}$ for some $j$, then $f_{\sigma(i)}$ ($i \in I_j^\sigma$)
can permuted counterclockwisely (equivalently, clockwisely) by Theorem \ref{theorem:main_theorem-x} (i), which  completes the proof. 
\end{proof}

This lemma directly implies
that an optimal permutation for linear functions $f_1, \dots , f_n$ can be computed in $O(k! \, n^{k+4})$ time, 
where $k$ denotes the number of monotone decreasing  $f_i$'s.

Assume first that no $f_i$ is identical
and we utilize $f_i^{(\epsilon)}$'s instead of $f_i$'s in Lemma \ref{lemma:epslion1}. 
By Lemma \ref{lemma:strong_L_left_R_right} (iii), we essentially have $n^2$ possible angles  $\psi_1$ and $\psi_2$. 
Based on such angles, 
we partition the set of indices of 
monotone linear functions into $I_0, \ldots, I_k$. 
By Lemma \ref{lemma:strong_L_left_R_right} (i), (ii), and (iv), 
we have  at most $n^{k+1}$ many such partitions.  
Since there exist $k!$ orderings of monotone decreasing functions, 
by checking at most $k! \,n^{k+3} (=n^2 \times n^{k+1} \times k!)$
permutations $\sigma$, 
we obtain an optimal permutation for $f_i$'s. 
Note that each such permutation $\sigma$ and the composite $f^\sigma$ can be computed in $O(n)$ time, after sorting  $\theta(f^{(\epsilon)}_i)$'s.  
Since $\theta(f_s^{(\epsilon)})$ and $\theta(f_t^{(\epsilon)})$ can be compared in $O(1)$ time for sufficiently small $\epsilon >0$ without exactly computing their angles, 
we can sort $\theta(f^{(\epsilon)}_i)$'s 
in $O(n\log n)$ time. 
Thus in total we require 
$O(k! \,n^{k+4}+n\log n)=O(k! \,n^{k+4})$
time, if no $f_i$ is identical. 
If some $f_i$'s are identical, then we can 
put them into $I_0$, where $I_0$ is obtained in the procedure above for the non-identical functions. Therefore, an optimal permutation can be computed in 
$O(k! \,n^{k+4})$
time. 

In order to improve this XP result, namely, to 
have an FPT algorithm with respect to $k$, 
we apply the dynamic programming approach to the following problem.  

\medskip

\noindent
{\bf Problem} {\sc LU-Ordered Optimal Composition}

\smallskip

\noindent
{\bf Input}: Two sets of monotone linear functions 
$L=\{p_1, \dots , p_{|L|}\} \ \text{ and } \ U= 
\{q_1, \dots , q_{|U|}\}$,
and monotone decreasing linear functions $g_1, \dots , g_{k}$ with $k>0$.
\smallskip

\noindent
{\bf Output}: An optimal permutation $\sigma$ for linear functions in $L \cup U \cup \{g_1, \dots , g_k\}$ such that 

\begin{description}
\item[{\rm (i)}]$L^\sigma=L$ and $U^\sigma=U$,
\item[{\rm (ii)}]the restriction of $\sigma$ on $L$ produces the ordering $(p_1, \dots , p_{|L|})$, 
\ and 
\item[{\rm (iii)}]
the restriction of $\sigma$ on $U$  produces the ordering $(q_1, \dots , q_{|U|})$. 
\end{description}
\noindent
  Note that an optimal permutation for the original problem can be computed by solving  {\bf Problem} {\sc LU-Ordered Optimal Composition} $O(n^4)$ times for $|L|+|U| \leq n-k$. 
  Since the problem can be solved in 
  $O(2^{k} k (|L|+|U|+k)^2)$ time, 
  we obtain the following result. 

\begin{theorem}\label{theorem:FPT}
An optimal permutation for linear functions $f_1, \dots , f_n$ can be computed in $O(2^k k n^6)$ time if $k>0$, 
where $k$ denotes the number of monotone decreasing  $f_i$'s. 
\end{theorem}

 We first show that an optimal permutation for the original problem can be computed by solving    
  {\bf Problem} {\sc LU-Ordered Optimal Composition} $O(n^4)$ times, 
where the formal description can be found in {\bf Algorithm} \ref{algo:FPTa}. 
  
\begin{lemma}\label{lemma:reduce_to_ST_order_given}
An optimal permutation for linear functions $f_1, \dots , f_n$ can be computed 
in $O(n \log n$ $+n^4 T^*)$ time, where $T^*$ denotes the time required to solve  {\bf Problem} {\sc LU-Ordered Optimal Composition}. 
\end{lemma}
\begin{proof}
By the discussion before the description of {\bf Problem} {\sc LU-Ordered Optimal Composition},  
we assume without loss of generality that 
no $f_i$ is constant and identical. 
In order to make use of {\bf Problem} {\sc LU-Ordered Optimal Composition}, 
we first arrange all monotone functions $f_{i_1}, \dots , f_{i_n}$ by their angles, and compose $f_{i_t} \circ \dots \circ f_{i_s}$ if $\theta(f_{i_s})=\dots = \theta(f_{i_t})$ by Lemmas \ref{lemma:sin} and \ref{lemma:strong_L_left_R_right}. 
Let $\check{f_1}, \ldots, \check{f}_{\check{n}}$ be the resulting monotone functions such that $\check{f_i}\not= \check{f_j}$ for distinct $i$ and $j$ in $[\check{n}]$. 
By Lemma \ref{lemma:strong_L_left_R_right} (iii), we consider at most $n^2$ many partitions $(L,U)$ of $[\check{n}]$. For each such partition,  Lemma \ref{lemma:strong_L_left_R_right} (i) (resp., (ii)) implies that $L$ (resp., $U$)
is permuted in at most $n$ counterclockwise (resp., clockwise) ways.
Namely, by solving  {\bf Problem} {\sc LU-Ordered Optimal Composition} at most $n^4\,(=n^2 \times n \times n)$ times, we can find an optimal function $(\check{f})^{\check{\sigma}}$, from which an optimal permutation $\sigma$ for $f_i$'s can be computed.  
These computation totally requires in $O(n \log n+n^4 T^*)$ time. 
\end{proof}

To apply a dynamic programming approach to  {\bf Problem} {\sc LU-Ordered Optimal Composition},
 for $s \in \{0,\dots,|L|\}$, $t \in \{0,\dots,|U|\}$, and $G\subseteq [k]$, let $E(s,t,G)$ denote
 the set of permutations $\sigma$ for  $p_i$ ($i=1, \dots , s$), $q_i$  ($i=1, \dots , t$), and $g_i$ ($i \in G$) such that 
\begin{description}
\item[{\rm (i)}]$L^\sigma$ satisfies 
\begin{equation*}
L^\sigma=
\left\{
\begin{array}{ll}
\{p_1, \dots , p_s\} &\text{if $k-|G|$ is even}\\
\{q_1, \dots , q_t\}&\text{if $k-|G|$ is odd},  \end{array}\right. 
\end{equation*}
\item[{\rm (ii)}]$U^\sigma$ satisfies 
\begin{equation*}
U^\sigma=
\left\{
\begin{array}{ll}
 \{q_1, \dots , q_t\}&\text{if $k-|G|$ is even}\\
\{p_1, \dots , p_s\}&\text{if $k-|G|$ is odd},  \end{array}\right. 
\end{equation*}

\item[{\rm (iii)}]
the restriction of $\sigma$ on $\{p_1, \dots , p_s\}$ produces the ordering $(p_1, \dots , p_{s})$, 
\ and 
\item[{\rm (iv)}]
the restriction of $\sigma$ on $\{q_1, \dots , q_t\}$  produces the ordering $(q_1, \dots , q_{t})$, 
\end{description}
and let $F[\sigma]$ denote the linear function obtained by a permutation $\sigma \in E(s,t,G)$. 
By definition, $E(s,t,G)=\emptyset$ if and only if either (I) $s=t=|G|=0$, (II) $G=\emptyset$, $t>0$, and $k$ is even, or (III) $G=\emptyset$, $s>0$, and $k$ is odd.  
Furthermore, we define $v(s,t,G)$ by 
 \begin{equation*}
 v(s,t,G)=\left\{
\begin{array}{ll}
\min\{ F[\sigma] \mid \sigma \in E(s,t,G) \} &\text{if $k-|G|$ is even}\\
\max\{ F[\sigma] \mid \sigma \in E(s,t,G) \}  &\text{if $k-|G|$ is odd}, 
 \end{array}
 \right.
 \end{equation*}
where we define 
$v(s,t,G)=x$ in Case (I) of $E(s,t,G)=\emptyset$, 
and 
$\bot$ in Cases (II) and (III) of $E(s,t,G)=\emptyset$. 
By definition, $v(|L|,|U|, [k])$ denotes the optimal function for {\bf Problem} {\sc LU-Ordered Optimal Composition}. 
We note that $v(s,t,G)$ satisfies the following recursion if  $E(s,t,G)\not=\emptyset$.  
 \begin{equation}
 \label{eq-lastlast1}
 v(s,t,G)=\left\{\!\!
\begin{array}{ll}
\min\bigl(\{ p_s \circ v(s-1,t,G)\} 
\cup \{ g \circ v(s,t,G\setminus \{g\}) \mid g \in G\}\bigr) &\text{ if $k-|G|$ is even},\\
\max \bigl(\{ q_t \circ v(s,t-1,G)\} 
\cup \{ g \circ v(s,t,G\setminus \{g\}) \mid g \in G\}\bigr) &\text{ if $k-|G|$ is odd}.  
 \end{array}
 \right.
 \end{equation}
{\bf Algorithm} \ref{algo:LU-ORDERED}
formally describes the dynamic programming approach for {\bf Problem} {\sc LU-Ordered Optimal Composition}. 
\begin{lemma}\label{theorem:FPTe}
  {\bf Problem} {\sc LU-Ordered Optimal Composition}
  can be  solved  in $O(2^k k n^2)$ time.
\end{lemma}

\begin{proof}
As discussed above, the function $v(|L|, |U |, [k])$ denotes the optimal function for {\bf Problem} {\sc LU-Ordered Optimal Composition}. 
By \eqref{eq-lastlast1}, 
we can apply a dynamic programming approach to the problem. 
Since  $v$ has $(|L|+1) \times (|U|+1) \times 2^k= O(2^k n^2)$ entries and each entry can be computed in $O(k)$ time, $v(|L|, |U |, [k])$ can be computed in  $O(2^k k n^2)$ time. 
Since the corresponding permutation can also be computed in the same amount of time, the proof is completed. 
    \end{proof}
\begin{algorithm}[htb]\small
    \caption{ to solve {\sc LU-Ordered Optimal Composition}}
    \label{algo:LU-ORDERED}
    \begin{algorithmic}[1]
    \STATE $v(0,0,\emptyset)\leftarrow$ identical function
    \FOR {$s=0$ to $|L|$}
    \FOR {$t=0$ to $|U|$}
    \FORALL {$G\subseteq [k]$}
    \IF{$\Sigma(s,t,G)=\emptyset$}
    \STATE \textbf{continue}
    \ENDIF
    \IF{$k-|G|$ is even}
        \STATE $v(s,t,G)\leftarrow +\infty$
        \FOR {$g\in G$}
        \STATE $v(s,t,G)\leftarrow\min\{v(s,t,G),g\circ v(s,t,G\setminus\{g\})\}$
        \ENDFOR
        \IF{$s>0$}
        \STATE $v(s,t,G)\leftarrow\min\{v(s,t,G),p_s\circ v(s-1,t,G)\}$
    \ENDIF
    \ELSE
        \STATE $v(s,t,G)\leftarrow -\infty$
        \FOR {$g\in G$}
        \STATE $v(s,t,G)\leftarrow\max\{v(s,t,G),g\circ v(s,t,G\setminus\{g\})\}$
        \ENDFOR
        \IF{$t>0$}
        \STATE $v(s,t,G)\leftarrow\max\{v(s,t,G),q_t\circ v(s,t-1,G)\}$
        \ENDIF
    \ENDIF
    \ENDFOR
    \ENDFOR
    \ENDFOR
    \RETURN $v(|L|,|U|,[k])$
    \end{algorithmic}
\end{algorithm}

\begin{proof}[Proof of Theorem \ref{theorem:FPT}]
It follows from Lemmas \ref{lemma:reduce_to_ST_order_given} and  \ref{theorem:FPTe}.  
\end{proof}

\begin{algorithm}[ht]\small
    \caption{to solve the composition ordering problem for general linear functions}
    \label{algo:FPTa}
    \begin{algorithmic}[1]
    \STATE $val\leftarrow +\infty$ and $F,G\leftarrow \emptyset$ \ 
  \ /*Regard $F$ and $G$ as multisets */ 
    \FOR {each $i=1, \dots , n$}
     \IF{$\alpha(f_i)> 0$ and $f_i\not=x$}
     \STATE $F \leftarrow F \cup \{f_i\}$
     \ELSIF{$\alpha(f_i)=0$}
     \STATE $F \leftarrow F \cup \{f^{(\epsilon)}_i\}$
     \ELSE
     \STATE $G \leftarrow G \cup \{f_i\}$
     \ENDIF
     \ENDFOR
    \STATE arrange functions in $F$ for their angles and compose functions with the same angle\ \ \\
    /* Denote by $\check{F}$ the set of resulting functions */
    \FOR {each $\psi_1 \in \{   \theta(\check{f}) \in (0,\pi) \mid \check{f} \in \check{F} \} \cup \{\pi\}$  and $\psi_2 \in  \{   \theta(\check{f}) \in (\pi,2\pi) \mid \check{f} \in \check{F} \} \cup \{\pi\}$}
    \STATE $L\leftarrow \{ \check{f} \in \check{F} \mid \psi_1 \leq \theta(\check{f}) \leq \psi_2\}$
    and  $U\leftarrow \check{F}\setminus L$
    \FOR{each counterclockwise permutation $\tau_L$ for $L$ and clockwise permutation $\tau_U$ for $U$}
    \STATE $val\leftarrow\min\{val,$ Algorithm \ref{algo:LU-ORDERED}
    $( L,U,G)\}$ \\   
    /* Here functions in $L$ and $U$ are assumed to be arranged according to $\tau_L$ and $\tau_U$ */
    \ENDFOR
    \ENDFOR
    \RETURN $val$
    \end{algorithmic}
\end{algorithm}

\section{Multiplication Ordering for Matrices} \label{section:matrix}
In this section, we 
consider matrix multiplication orderings as a generalization of
composition orderings for linear functions. 
Recall that 
the problem  is to 
 find a permutation $\sigma:[n]\rightarrow[n]$ that minimizes $\bm{w}^\top \Mplr{n}{1} \bm{y}$ for given  $n$ matrices $M_1,\dots,$ $M_n\in \Br^{m\times m}$ and two vectors $\bm{w},\bm{y}\in \Br^m$, 
 where $m$ denotes a positive integer. 
As mentioned in the introduction, 
if we set $\bm{w}=\Mto{1}{0}
$,  $\bm{y}=\Mto{0}{1}$, and $M_i=\Mtt{a_i}{b_i}{0}{1}$ for any $i \in [n]$, 
then the matrix multiplication ordering problem is equivalent to the composition ordering problem for linear functions $f_i(x) = a_i x+ b_i$. 
We  show that the results for linear functions can be extended to the matrix multiplication for $m=2$. 
Furthermore, by applying max-plus algebra, 
we obtain the result in \cite{MaxPlus}, which is an extension of  Johnson's rule \cite{Johnson} for the two-machine flow shop scheduling,  as a corollary of our result.
Finally we show that possible generalizations of the problem turn out to be intractable, unless P$=$NP. 

\subsection{Matrices in linear algebra}
In order to prove Theorem \ref{theorem:triangularizable}, 
we assume that matrices $M_i$ in $\Br^{2\times 2}$ are all upper triangular, i.e., $M_i=\Mtt{a_i}{b_i}{0}{d_i}$, and claim that the problem can be reduced in linear time to the problem of the minimum multiplication problem for upper triangular matrices with nonnegative 
(2,2)-entry, $\bm{w}^{\top}=\Mot{1}{0}$, and $\bm{y}^{\top}=\Mot{0}{1}$.  

Let $p$ denote the number of matrices $M_i$ with negative (2,2)-entry.  
For a matrix $M_i$, define a matrix $N_i$ by $M_i$ if $d_i \geq 0$, and $-M_i$ otherwise. 
For a permutation $\sigma:[n]\to [n]$, we have 
 \begin{align*}
     \Mot{w_1}{w_2} N^\sigma \Mto{y_1}{y_2}
     &= w_1 y_1 (N^\sigma)_{1,1} + w_2 y_2(N^\sigma)_{2,2} + w_1 y_2 (N^\sigma)_{1,2}\\
     &= (-1)^p \left(w_1y_1 \prod_{i=1}^n a_i + w_2y_2\prod_{i=1}^n d_i + w_1y_2 \Mot{1}{0} \Mall \Mto{0}{1}\right), 
 \end{align*}
where $\bm{w}^{\top}=\Mot{w_1}{w_2}$ and $\bm{y}^{\top}=\Mot{y_1}{y_2}$.  
Since $p$, 
$w_1y_1 \prod_{i=1}^n a_i$,  $w_2y_2\prod_{i=1}^n d_i$, and $w_1y_2$ are constant, 
 it is enough to consider the $(1, 2)$-th entry of $\Mall$, i.e., $\Mot{1}{0} \Mall \Mto{0}{1}$. 
Moreover, if $w_1y_2=0$ then any permutation is optimal.
 Therefore, we assume that $w_1y_2\neq 0$, and  
 by the following lemma, 
 we only need to examine the minimization.  
 This completes the claim. 

For a $2 \times 2$ matrix $M=\Mtt{a}{b}{0}{d}$, 
define $\tilde{M}=\Mtt{a}{-b}{0}{d}$.  
 \begin{lemma}\label{lemma:tilde_matrix}
     For any permutation $\sigma:[n]\to [n]$, we have 
     \begin{align*}
  \Mot{1}{0} \Mall \Mto{0}{1} = - \Mot{1}{0} \tilde{M}_{\sigma(n)}\cdots\tilde{M}_{\sigma(1)} \Mto{0}{1}.
     \end{align*}
 \end{lemma}
\begin{proof}
It follows from the definition of $\tilde{M}$.
\end{proof}

We next transform the problem into 
the composition ordering problem for linear functions. 
     For $M=\Mtt{a}{b}{0}{d}$ and a real number $\epsilon \neq 0$, we define 
     \begin{align*}
  M^{(\epsilon)}=
  \begin{cases}
      M & (d\neq 0),\\
      M + \Mtt{0}{0}{0}{\epsilon} & (d=0).\\
  \end{cases}
     \end{align*}

 Similar to Lemma \ref{lemma:epsilon}, we have the following lemma.
 \begin{lemma}
   For $n$ triangular matrices $M_1, \dots, M_n \in \Br^{2\times 2}$, there exists a real number $r >0$ such that
$  \Mot{1}{0} (M^{(\epsilon)})^{\sigma} \Mto{0}{1} \leq
\Mot{1}{0} (M^{(\epsilon)})^{\rho} \Mto{0}{1}$ implies
$  \Mot{1}{0} \Mall \Mto{0}{1} \leq \Mot{1}{0} \Mall[\rho]\Mto{0}{1}$
 for any two permutations $\sigma, \rho$ and any $\epsilon$ with $|\epsilon| < r$.   
 \end{lemma}
 
 \begin{proof}
 The proof is similar to that of Lemma \ref{lemma:epsilon}, and is left to the reader.
 \end{proof}
 

Thus we assume that given upper triangular matrices   $M_i=\Mtt{a_i}{b_i}{0}{d_i}$  have positive  $d_i$. We then have  
 \begin{align*}
     \Mot{1}{0}\Mall\Mto{0}{1}
     \ &= \ \left(\prod_{i=1}^n d_i \right) \Mot{1}{0} \Mtt{\dfrac{\Ap{n}}{\Dp{n}}}{\dfrac{\Bp{n}}{\Dp{n}}}{0}{1} \cdots \Mtt{\dfrac{\Ap{1}}{\Dp{1}}}{\dfrac{\Bp{1}}{\Dp{1}}}{0}{1} \Mto{0}{1}\\
     &=  \ \left(\prod_{i=1}^n d_i \right) \Fall(0),
 \end{align*}
 where we define 
 \begin{align*}
     f_i(x) \,\,= \,\,\frac{a_i}{d_i}x+\frac{b_i}{d_i} \ \ \ \mbox{ for } i \in [n]. 
 \end{align*}
This implies that matrix multiplication for upper triangular matrices can be solved by solving the composition ordering problem for linear functions. 

We remark that our algorithm concerns the comparison of  polar angles $\theta(f_i)$'s, but not of the vectors $\Mto{b_i/d_i}{1-a_i/d_i}$, and hence we do not need to care about the case when $d_i=\epsilon$.     
 Therefore, we have the following lemma for $2 \times 2$ triangular matrices.
 \begin{lemma}\label{lemma:triangular_basic}
     For  upper triangular matrices $M_1, \dots, M_n$ in $\Br^{2 \times 2}$, we have the following statements. 
     \begin{enumerate}
  \item[{\rm (i)}] If all matrices have nonnegative  determinants, then an optimal multiplication ordering can be computed in $O(n\log n)$ time.
  \item[{\rm (ii)}] If some matrix has negative determinant, then an optimal multiplication ordering can be computed in $O(k2^k n^6)$ time, where $k$ denotes the number of matrices with negative determinants. 
     \end{enumerate}
 \end{lemma}

This immediately implies Theorem  \ref{theorem:triangularizable}. 
\begin{proof}[Proof of Theorem \ref{theorem:triangularizable}]
By Lemma \ref{lemma:triangular_basic}, 
we only reduce the problem to the one for upper triangular matrices.
 
 Let $M_1, \ldots, M_n$ be $2 \times 2$ simultaneously triangularizable matrices. 
 Since there exists a regular matrix $P\in\Br^{2\times 2}$ such that $P^{-1}M_iP$ is an upper triangular matrix $T_i$ for any $i\in[n]$, we have
 \begin{align*}
     \bm{w}^\top \Mall \bm{y}
  &= \bm{w}^\top P (P^{-1} M_{\sigma(n)}P) (P^{-1}M_{\sigma(n-1)}P) \cdots (P^{-1} M_{\sigma(1)} P) P^{-1}\bm{y}\\
  &= \bm{w'}^\top T^{\sigma} \bm{y'},
 \end{align*}
 where $\bm{w'}=P^\top\bm{w}$ and $\bm{y'}=P^{-1}\bm{y}$.
 Thus we can reduce the problem to the one for triangular matrices.
  \end{proof}

Unfortunately, this positive results cannot be extended to 
1) the nonnegative determinant case for $m=2$, 
2) the case of $m \geq 3$ and 
3) the target version; See Theorems \ref{theorem:matrix-hardness1} (i), (ii) and \ref{theorem:minimize_norm-hardness}.

\begin{proof}[Proof of Theorem \ref{theorem:matrix-hardness1} (i)]
     We show that {\sc 3-Partition} can be reduced to the problem, where {\sc 3-Partition} is,  given $n(=3m)$ positive integers $a_1,\dots,a_n$ with $\sum_{i=1}^n a_i = mT$, to decide if  there exist $m$ disjoint sets $P_1,\dots,P_m \subseteq [n]$ such that $|P_j|=3$ and $\sum_{i\in P_j} a_i=T$ for all $j \in [m]$.

  The problem is strongly NP-complete even if each $a_i$ satisfies $\frac{T}{4}<a_i<\frac{T}{2}$.
 Hereafter we assume such condition.
 It follows from this assumption that $\sum_{i\in P} a_i = T$ only when $|P|=3$.
 Thus we do not have to care about the constraint $|P_j|=3$.
 
Given an instance of {\sc 3-Partition}, 
we construct $n+m-1$ matrices $M_i\in \Br^{2\times 2}$ and $2$ vectors $\bm{w},\bm{y}\in \Br^2$ as follows:
     \begin{align*}
         M_i&=
         \begin{cases}
             \Mtt{1}{a_i}{0}{1} & \text{if } i=1,\dots,n\\
             \Mtt{0}{0}{1}{0} & \text{if } i=n+1,\dots,n+m-1, 
         \end{cases} \,
         \bm{w} = \Mto{1}{0},  \mbox{ and }\ 
         \bm{y} = \Mto{0}{-1}.
     \end{align*}
     It is easy to see that $\det(M_i)\geq 0$ for all $M_i$'s.
     We claim that $-T^m$ is the optimal value for the optimal multiplication ordering problem if and only if there exists a desired partition $P_1,\dots,P_m$ 
    of the corresponding instance of  {\sc 3-Partition}, which completes the proof. 
 For a permutation $\sigma:[n+m-1]\rightarrow[n+m-1]$, 
define positive integers $\ell_1<\dots<\ell_{m-1}$ by $\sigma(\ell_j) \in [n+1,n+m-1]$, i.e., 
 \begin{align*}
         M_{\sigma(\ell_j)}&= \Mtt{0}{0}{1}{0}
     \end{align*}
for $j\in[m-1]$, and let 
$\ell_0=0$ and $\ell_m=n+m$. 
Let  
\[P_j = \{i\in[n+m-1] \mid \ell_{j-1}<i<\ell_j\} \ \ \mbox{ for } j\in[m].
\]
     Then we have
     \begin{align*}
         M^\sigma 
         &= \prod_{i\in P_m} M_{\sigma(i)} \Mtt{0}{0}{1}{0} \prod_{i\in P_{m-1}}M_{\sigma(i)}  \cdots \prod_{i\in P_2}M_{\sigma(i)} \Mtt{0}{0}{1}{0} \prod_{i\in P_1}M_{\sigma(i)}. 
     \end{align*}
Note that $\bm{w}^{\rm T}M^\sigma
         \bm{y}=0$ if some $P_j$ is empty. 
  On the other hand, if all $P_j$'s are nonempty, then  $M^\sigma$ can be  restated as follows.           
  \begin{align*}
         M^\sigma 
         &= \Mtt{1}{\displaystyle \sum_{i\in P_m} a_{\sigma(i)}}{0}{1} \Mtt{0}{0}{1}{0} \Mtt{1}{\displaystyle \sum_{i\in P_{m-1}} a_{\sigma(i)}}{0}{1} \cdots \Mtt{0}{0}{1}{0} \Mtt{1}{\displaystyle \sum_{i\in S'_1} a_{\sigma(i)}}{0}{1}\\
         &= \Mtt{1}{\displaystyle \sum_{i\in P_m} a_{\sigma(i)}}{0}{1} \Mtt{0}{0}{1}{\displaystyle \sum_{i\in P_{m-1}} a_{\sigma(i)}} \cdots \Mtt{0}{0}{1}{\displaystyle \sum_{i\in P_1} a_{\sigma(i)}}\\
         &= \Mtt{1}{\displaystyle \sum_{i\in P_m} a_{\sigma(i)}}{0}{1} \Mtt{0}{0}{\displaystyle\prod_{j\in[2,m-1]}\sum_{i\in P_j}a_{\sigma(i)}}{\displaystyle\prod_{j\in[m-1]}\sum_{i\in P_j}a_{\sigma(i)}}\\
         &= \Mtt{\displaystyle \prod_{j\in[2,m]}\sum_{i\in P_j}a_{\sigma(i)}}{\displaystyle\prod_{j\in[m]}\sum_{i\in P_j}a_{\sigma(i)}}{\displaystyle\prod_{j\in[2,m-1]}\sum_{i\in P_j}a_{\sigma(i)}}{\displaystyle\prod_{j\in[m-1]}\sum_{i\in P_j}a_{\sigma(i)}}, 
     \end{align*}
which implies that   $\bm{w}^\top M^\sigma \bm{y} = -\prod_{j\in[m]}\sum_{i\in P_j}a_{\sigma(i)}$.
Therefore, it is regarded as the problem of computing a partiotion $P_1, \dots , P_m$ of $[n]$ 
with the minimum  $-\prod_{j\in[m]} \sum_{i\in P_j} a_i$.
     By the inequality of arithmetic and geometric means, we have $-\prod_{j\in[m]} \sum_{i\in P_j} a_i\geq -T^m$, where the equality holds if and only if 
      $\sum_{i\in P_j} a_i = T$ for all $j\in[m]$.
      Since each $a_i$ satisfies $\frac{T}{4}<a_i<\frac{T}{2}$, $\sum_{i\in P_j} a_i = T$ implies $|P_j|=3$. 
     This proves the claim.     
\end{proof}

In order to prove Theorem \ref{theorem:minimize_norm-hardness}, we next consider the problem of computing a permutation $\sigma$ of monotone linear functions $f_1, \dots , f_n$ that minimizes $|\beta(f^\sigma)-t|$ for a given target  $t \in \Br$. 
 

We again reduce {\sc 3-Partition} 
to the problem. 
Given an instance of {\sc 3-Partition}, 
We construct $n+m\,(=4m)$ linear functions $f_i$ and a target $t \in \Br$ as follows:
     \begin{align*}
         f_i(x)&=
         \begin{cases}
            x+a_i & \text{if } i=1,\dots,n\\
            (c_2+mT)(x-(i-n)T)+(i-n)T & \text{if } i=n+1,\dots,n+m\\
         \end{cases}\\
         t&=mT,
     \end{align*}
Note that all $f_i$'s are monotone, and the following statements hold. 
\begin{align}
&\mbox{any } i\in[n+m] \mbox{ satisfies }
f_{i}(x)>x  \mbox{ if } x> mT\label{eq--xs1}\\
&\mbox{any } j\in[m] \mbox{ satisfies }f_{n+j}(jT)=jT  
\mbox{ and }
f_{n+j}(x)<x  \mbox{ if } x<0 \label{eq--xs2}
\end{align}
 Furthermore, we have the following lemmas. 
 \begin{lemma}
 \label{lemma-npc--a1}
Let $f_1, \dots , f_{n+m}$, and $t$ be defined as above. 
Then there exists a partition $
P_1,\dots,$ $P_m\subseteq[n]$ such that $\sum_{i\in P_j} a_i=T$ for $j\in[m]$ if and only if there exists a permutation $\sigma:[n+m]\rightarrow[n+m]$ such that  $\Fplr{\sigma^{-1}(n+j)-1}{1}(0)=jT$ for any $j\in[m]$.
\end{lemma}

\begin{proof}
     To prove only if part, let 
     $\{P_j=\{p_{j,1},p_{j,2},p_{j,3}\} \mid j\in[m]\}$ be a desirable partition of $[n]$, i.e., 
     $a_{p_{j1}}+a_{p_{j2}}+a_{p_{j3}}=T$ for any $j \in [m]$. 
      We define a permutation $\sigma:[n+m]\rightarrow[n+m]$ by 
     \begin{align}
         \sigma(i)=
         \begin{cases}
             p_{j,k} & \text{if } i=4(j-1)+k \text{ with } j\in[m],k\in[3],\\
             n+j & \text{if } i=4j \text{ with } j\in[m].
         \end{cases}\label{eq-importa1}
     \end{align}
For any $j\in[m]$, we have 
     \begin{align*}
         \Fplr{4j}{4j-3} 
         &\,=\, f_{n+j}\circ(x+a_{p_{j,3}})\circ(x+a_{p_{j,2}})\circ(x+a_{p_{j,1}})\\
         &\,= \,f_{n+j}\circ(x+T).
     \end{align*}
Therefore,  by the induction on $j$, it is not difficult to see that $\Fplr{4j-1}{1}(0)=jT$ for $j\in[m]$, which completes the only-if part, 
since $\sigma^{-1}(n+j)=4j$. 

To prove if part, assume that  $\sigma$ satisfy $\Fplr{\sigma^{-1}(n+j)-1}{1}(0)=jT$ for any $j\in[m]$.
Define $v_i=\Fplr{i}{1}(0)$ for any $i\in[n+m]$ and $v_0=0$.
Then we can see that  $\sigma(i)\in[n]$ implies  $v_i=v_{i-1}+a_{\sigma(i)}$. 
     If $\sigma(i)=n+j$ for some $j\in[m]$ then we have $v_i=v_{i-1}$ by \eqref{eq--xs2} and the assumption of $\sigma$.
    Thus $v_1,\dots,v_{n+m}$ are nondecreasing.
     Let $\ell_1,\dots,\ell_m\in[n+m]$ satisfy $\sigma(\ell_j)=n+j$ for $j\in[m]$.
     Then they are increasing, 
     since $v_1,\dots,v_{n+m}$ are nondecreasing.
     Thus again by \eqref{eq--xs2} and the assumption of $\sigma$, $\{P_j=\{\sigma(i) \mid \ell_{j-1}<i<\ell_j\} \mid j\in[m] \}$ is a desirable partition of  $[n]$, where $\ell_0=0$. 
\end{proof}

 \begin{lemma}
 \label{lemma-npc--a2}
Let $f_1, \dots , f_{n+m}$, and $t$ be defined as above. 
If there exists a partition $
P_1,\dots,P_m\subseteq[n]$ such that $\sum_{i\in P_j} a_i=T$ for $j\in[m]$, then $\min_\rho |\beta(f^\rho)-t|=0$ holds. Otherwise, we have $\min_\rho |\beta(f^\rho)-t| > c_2$. 
\end{lemma}

\begin{proof}
Assume that a desirable partition $P_1, \dots , P_m$ exists. Then as shown in the proof of Lemma \ref{lemma-npc--a2}, 
a permutation given in \eqref{eq-importa1} provides the composite $f^\sigma$ such that  $\beta(f^\sigma)=mT\,(=t)$, which proves the first half of the lemma. 

On the other hand, if no desirable partition exists. Then  by Lemma \ref{lemma-npc--a2}, 
 any permutation $\sigma$ has some $j\in[m]$ such that $\Fplr{\sigma^{-1}(n+j)-1}{1}(0)\neq jT$.
     We separately consider the following two cases.

{\bf Case 1.} If $\Fplr{\sigma^{-1}(n+j)-1}{1}(0)=jT+\Delta$ for some $\Delta\in\mathbb{Z}_{>0}$, then we have 
     \begin{align*}
         \Fplr{\sigma^{-1}(n+j)}{1}(0) 
         &= f_{n+j}(jT+\Delta)\\
         &= (c_2 +mT)\Delta+jT
         > mT+c_2.
     \end{align*}
     By \eqref{eq--xs1}, 
    $f^\sigma(0)>\Fplr{\sigma^{-1}(n+j)}{1}(0)>mT+c_2$, implying that  $|f^\sigma(0)-t|>c_2$.

     {\bf Case 2.} If $\Fplr{\sigma^{-1}(n+j)-1}{1}(0)=jT-\Delta$ for some $\Delta\in\mathbb{Z}_{>0}$, 
     then we have  
     \begin{align*}
         \Fplr{\sigma^{-1}(n+j)}{1}(0) 
         &= f_{n+j}(jT-\Delta)\\
         &= -(c_2+mT)\Delta+jT
         <-c_2.
     \end{align*}
   By \eqref{eq--xs2} and $f_{i}(x)=x+a_i$ for any $j\in[n]$,  we can show that $f^\sigma(0) < -c_2+mT $, which implies  $|f^\sigma(0)-mT|>c_2$.

 In either case, we can prove $|f^\sigma(0)-mT|>c_2$, completing the proof. 
 \end{proof}

\begin{proof}[Proof of Theorem \ref{theorem:minimize_norm-hardness}]
This follows from Lemma \ref{lemma-npc--a2}.
\end{proof}


Theorem \ref{theorem:matrix-hardness1} (ii) will be proved in the next subsection since we reduce an NP-hard problem expressed using the max-plus algebra to the problem.

\subsection{Matrices in the max-plus algebra}\label{App.maxplus}
    In this section, we investigate multiplication of matrices in the max-plus algebra. 
    We obtain the result of Bouquard et al. in the case $m=2$ as a corollary of our result, and then prove Theorem \ref{theorem:minimize_norm-hardness}. 

Let $\Br_{\max}$ be the set $\Br \cup \{-\infty\}$ with two binary operations max and + denoted by $\oplus$ and $\otimes$ respectively, i.e., for $a,b\in\Br_{\max}$,
\begin{gather*}
    a\oplus b = \max\{a,b\}  \ \text{ and }\  a\otimes b = a+b.
\end{gather*}
$(\Br_{\max},\oplus,\otimes)$ is called the \textit{max-plus algebra}.
We denote by $\mathbb{0}$ the additive identity $-\infty$, and denote by $\mathbb{1}$ the multiplicative identity $0$.

Since the two operations $\oplus$ and $\otimes$ are extended to the matrices of $\Br_{\max}$ as in the  linear algebra, we can consider the problem to find a permutation $\sigma$ that minimizes/maximizes
$\bm{u}^\top \otimes N_{\sigma(n)} \otimes\cdots\otimes N_{\sigma(1)} \otimes \bm{v}$, where $N_i \in \Br_{\max}^{m \times m}$ and $\bm{u}, \bm{v}\in \Br_{\max}^m$.
We denote $N_{\sigma(n)} \otimes\cdots\otimes N_{\sigma(1)}$ by $N^{\sigma}$.




Bouquard, Lenté and Billaut \cite{MaxPlus} dealt with the problem to minimize the objective value
\begin{equation}\label{eq-maxtriangular}
\begin{pmatrix}
\mathbb{1} & \mathbb{0} & \ldots & \mathbb{0}
\end{pmatrix}
\otimes N_{\sigma(n)} \otimes\cdots\otimes N_{\sigma(1)} \otimes 
 \begin{pmatrix*}
     \mathbb{0}\\
     \vdots\\
     \mathbb{0}\\
     \mathbb{1}
 \end{pmatrix*},
 \end{equation}
 where each $N_i$ is an upper triangular matrix in $\mathbb{R}_{\max}^{m\times m}$.
They showed that the problem in the case $m=2$ is a  generalization of the two-machine flow shop scheduling problem to minimize the makespan, and is solvable in $O(n\log n)$ time by using an extension of Johnson's rule \cite{Johnson} for the two-machine flow shop scheduling  
We show that this result is obtained as a corollary of our result. 
by reducing their problem to the optimal multiplication ordering problem for $2 \times 2$ triangular matrices in linear algebra. 
The derivation can be found in Appendix  \ref{App.maxplus_m=2}.
They also proved that the problem in the case $m=3$ is strongly NP-hard by reduction from the three-machine flow shop scheduling problem to minimize the makespan, which is known to be strongly NP-hard \cite{Garey1976}.
We make use of this result to prove Theorem \ref{theorem:matrix-hardness1} (ii). 

In preparation, we show an important relation to matrices in linear algebra.
\begin{lemma}\label{lemma:maxplus_linear}
Let $\bm{u}^\top = (u_1, \ldots, u_m)$, $\bm{v}^\top = (v_1, \ldots, v_m)$, and 
 $N_1, \dots , N_n$ be matrices in $\Br_{\max}^{m \times m}$.  
For a positive real $\gamma$, 
define  $\bm{w}^\top = (\gamma^{u_1}, \ldots, \gamma^{u_m})$, $\bm{y}^\top = (\gamma^{v_1}, \ldots, \gamma^{v_m})$, 
and matrices $M_i$, $i \in [n]$, in  $\Br^{m \times m}$ by 
 $(M_i)_{jk}= \gamma^{(N_i)_{jk}}$. 
Then
\[
\bm{w}^\top M^\sigma \bm{y}
= \sum_{j,k=1}^m \sum_{t \in V[\sigma]_{jk}}\gamma^{u_j + t + v_k},
\]
where $V[\sigma]_{jk}$ denotes the arguments $($standard sums$)$ of the max operation  of $(N^\sigma)_{jk}$.
\end{lemma}
\begin{proof}
Straightforward by induction on $n$.
\end{proof}

\begin{example}\label{ex:maxplus22}
For  $\bm{u}^\top = (u_1, u_2)$, $\bm{v}^\top = (v_1, v_2)$, and 
$N_i = \Mtt{a_i}{b_i}{\mathbb{0}}{d_i}\in \Br_{\max}^{2 \times 2}$ for $i \in [n]$,  let 
$\bm{w}^\top = (\gamma^{u_1}, \gamma^{u_2})$, $\bm{y}^\top = (\gamma^{v_1}, \gamma^{v_2})$, and 
$M_i = \Mtt{\gamma^{a_i}}{\gamma^{b_i}}{0}{\gamma^{d_i}}$. 
Then we have
\begin{align*}
\bm{w}^\top M^\sigma \bm{y}
&= \bm{w}^\top\Mtt{\prod_{i}\gamma^{a_i}}{\sum_{i}\gamma^{b_{\sigma(i)}}\prod_{j_1<i}\gamma^{d_{\sigma(j_1)}}\prod_{j_2>i}\gamma^{a_{\sigma(j_2)}}}{0}{\prod_{i}\gamma^{d_i}}\bm{y}\\
&= \bm{w}^\top\Mtt{\gamma^{\sum_{i}a_i}}{\sum_{i}\gamma^{b_{\sigma(i)}+\sum_{j_1<i}d_{\sigma(j_1)}+\sum_{j_2>i}a_{\sigma(j_2)}}}{0}{\gamma^{\sum_{i}d_i}}\bm{y}\\
&= \gamma^{u_1+\sum_{i}a_i+v_1}+\sum_{i=1}^n\gamma^{u_1+b_{\sigma(i)}+\sum_{j_1<i}d_{\sigma(j_1)}+\sum_{j_2>i}a_{\sigma(j_2)}+v_2} +\gamma^{u_2 +\sum_{i}d_i+v_2}.
\end{align*}
On the other hand, we have
\begin{align*}
N^\sigma
&= \Mtt{a_1 \otimes \cdots \otimes a_n}{\bigoplus_{i}b_{\sigma(i)}\bigotimes_{j_1<i}d_{\sigma(j_1)}\bigotimes_{j_2>i}a_{\sigma(j_2)}}{\mathbb{0}}{d_1 \otimes \cdots \otimes d_n}\\
&= \Mtt{\sum_i a_i}{\max_{i}\left[b_{\sigma(i)} +\sum_{j_1<i}d_{\sigma(j_1)}+\sum_{j_2>i}a_{\sigma(j_2)}\right]}{\mathbb{0}}{\sum_i d_i }.
\end{align*}
Therefore, we obtain $V[\sigma]_{11}=\{\sum_i a_i\}$, $V[\sigma]_{12}=\{b_{\sigma(i)} +\sum_{j_1<i}d_{\sigma(j_1)}+\sum_{j_2>i}a_{\sigma(j_2)} \mid i \in [n]\}$, $V[\sigma]_{21}=\{\mathbb{0}\}$, and $V[\sigma]_{22}=\{\sum_i d_i\}$.
\end{example}

\subsubsection{Derivation of the result of Bouquard et al. for $m=2$ }\label{App.maxplus_m=2}
We reduce the problem of \eqref{eq-maxtriangular} to the multiplication ordering problem for $2 \times 2$ triangular matrices in linear algebra.

\begin{lemma}\label{lemma:reduce_to_scheduling}
For an index $i\in[n]$, let $N_i=\Mtt{a_i}{b_i}{\mathbb{0}}{d_i}\in\Br_{\max}^{2\times 2}$  with $a_i,b_i,d_i\neq \mathbb{0}$, and $M_i$ be a matrix in 
 $\Br^{2\times 2}$ defined as in Lemma \ref{lemma:maxplus_linear}, i.e., $M_i=\Mtt{\gamma^{a_i}}{\gamma^{b_i}}{0}{\gamma^{d_i}}$
 for a positive number $\gamma$. 
 Then there exists an $\Omega\in\Br$ such that 
 $\Mot{1}{0}M^{\sigma}\Mto{0}{1} \leq \Mot{1}{0}M^{\rho}\Mto{0}{1}$
 implies 
 $\Mot{\mathbb{1}}{\mathbb{0}}\otimes N^{\sigma} \otimes\Mto{\mathbb{0}}{\mathbb{1}} \leq \Mot{\mathbb{1}}{\mathbb{0}}\otimes N^{\rho} \otimes\Mto{\mathbb{0}}{\mathbb{1}}$
  for any two permutations $\sigma, \rho:[n]\to [n]$ and any $\gamma$ with $\gamma>\Omega$. 
     \end{lemma}
     
    \begin{proof}
    By Example \ref{ex:maxplus22}, we have
    \begin{align*}
     \Mot{1}{0}M^\sigma\Mto{0}{1}
     &= \sum_{t \in V[\sigma]_{12}} \gamma^{t},\\
     \Mot{\mathbb{1}}{\mathbb{0}}\otimes N^{\sigma} \otimes\Mto{\mathbb{0}}{\mathbb{1}}
     &= \max V[\sigma]_{12},
 \end{align*}
 where $V[\sigma]_{12}=\{b_{\sigma(i)} +\sum_{j_1<i}d_{\sigma(j_1)}+\sum_{j_2>i}a_{\sigma(j_2)} \mid i \in [n]\}$.
 Let us denote $\max V[\sigma]_{12}$ simply by $v(N^\sigma)$.
     If $v(N^\sigma) \leq v(N^\rho)$, then the statement clearly holds. We thus consider the case in which $v(N^\sigma) > v(N^\rho)$.  
 Define $\omega$ and $\Omega$ by
 \begin{align*}
     \omega= \min\{|v(N^\sigma)-v(N^\rho)| \colon v(N^\sigma) \neq v(N^\rho)\} \ \mbox{ and } \ 
     \Omega= n^{\frac{1}{\omega}}.
 \end{align*}
Let us  fix $\gamma>\Omega$. Since $\gamma^{\omega} > n$, we have
$\gamma^{v(N^\rho)+\omega} > n\gamma^{v(N^\rho)}$.
For any two permutations $\sigma, \rho$ such that $v(N^\sigma) < v(N^\rho)$,
the inequality $v(N^\sigma)+\omega < v(N^\rho)$ holds. 
 Moreover, the following two inequalities hold:
 \begin{align*}
     \sum_{t \in V[\sigma]_{12}} \gamma^{t} \leq n \gamma^{v(N^\sigma)} \  \mbox{ and }
      \ \gamma^{v(N^\rho)} \leq \sum_{t \in V[\rho]_{12}} \gamma^{t}.
 \end{align*}
 Combining the four inequalities, we obtain
 \begin{align*}
     \sum_{t \in V[\sigma]_{12}} \gamma^{t}
     < \sum_{t \in V[\rho]_{12}} \gamma^{t}.
 \end{align*} 
      \end{proof}

    Since $\gamma^{d_i} > 0$ for any $i \in [n]$, we have
    \begin{align*}
        \Mot{1}{0}\Mall\Mto{0}{1}
        &= \left(\prod_{i=1}^n \gamma^{d_i} \right) \Fall(0),
    \end{align*}
    where
    \begin{align*}
        f_i(x) = \gamma^{a_i - d_i}x+\gamma^{b_i - d_i}.
    \end{align*}
    It is guaranteed that $\theta(f_i)\in \left(-\frac{\pi}{2},\frac{\pi}{2}\right)_{2\pi}$ for any $i\in[n]$.
    Thus a permutation $\sigma$ is 
    minimum in both problems if $\frac{1 - \gamma^{\Ap{i} - \Dp{i}}}{\gamma^{\Bp{i}-\Dp{i}}}$ ($i=1, \dots ,n$) is nondecreasing. 

    Note that the length of $M_i$ is exponential of the length of $N_i$.
    Thus polynomial solvability  for the problem in the max-plus algebra does not immediately follow from the one in linear algebra.
    However, in this case, we can use the following index $\kappa$, 
    instead of computing $M_i$ explicitly.
    For a max-plus matrix $N=\Mtt{a}{b}{\mathbb{0}}{d}$, we define 
    \begin{align*}
 \kappa(N)=
 \begin{cases}
     (-1,b-a,d-b) & (a>d),\\
     (0,0,0) & (a=d),\\
     (1,d-b,a-b) &(a<d).
 \end{cases}
    \end{align*}
    It is not difficult to see that the following equivalence holds for a sufficiently large $\gamma$ 
    \begin{align*}
 \kappa(N_i)\,\preceq \,\kappa(N_j) 
 \ \Longleftrightarrow \ 
 \frac{\gamma^{d_i}-\gamma^{a_i}}{\gamma^{b_i}}\, \leq \,\frac{\gamma^{d_j}-\gamma^{a_j}}{\gamma^{b_j}},
    \end{align*}
    where $\preceq$ denotes the lexicographic order of $\kappa(N_i)$'s. 
    Therefore we can solve the problem in the max-plus algebra by simply sorting $\kappa(N_1),\dots,\kappa(N_n)$ in the lexicographic order.
 In fact, we do not have to take the third entry of $\kappa$ into account by the following lemma. 

\begin{lemma}\label{lemma:maxcommuting}
For an $i \in [m]$, let $N_i =\Mtt{a_i}{b_i}{\mathbb{0}}{d_i}$ with $a_i, b_i$, and $d_i$ in $\Br_{\max} \setminus \{\mathbb{0}\}$. 
Then the operation $\otimes$ is commutative on  $\{N_1, \dots, N_m\}$ if and only if
    \begin{enumerate}
      \item[{\rm (i)}] $a_i \geq d_i$ for any $i \in [m]$ and there exists a constant $c \in \Br$ such that {\rm (i-1)}  $c = b_k - a_k$ for any $k$ with $a_k > d_k$
      and {\rm (i-2)} $c \geq b_k - a_k$ for any $k$ with $a_k = d_k$,
      \item[or]
      \item[{\rm (ii)}] $a_i \leq d_i$ for any $i \in [m]$ and there exists a constant  $c \in \Br$ such that {\rm (ii-1)} $c = d_k - b_k$ for any $k$ with $a_k < d_k$ 
      and {\rm (ii-2)} $c \leq d_k - b_k$ for any $k$ with $a_k = d_k$.
   \end{enumerate}
  \end{lemma}

\begin{proof}
    Let $N_1$ and $N_2$ be commuting matrices, i.e., $N_2 \otimes N_1 = N_1 \otimes N_2$. This means that
    \begin{equation}\label{equation:Maxcommute}
    a_2 \otimes b_1 \oplus b_2 \otimes d_1 = a_1 \otimes b_2 \oplus b_1 \otimes d_2.
  \end{equation}
  Consider the following six cases:
\begin{itemize}
  \item Case 1: $a_1 > d_1$ and $a_2 > d_2$. Then we have $a_1 \otimes b_2 > b_2 \otimes d_1$ and $a_2 \otimes b_1 > b_1 \otimes d_2$. Hence (\ref{equation:Maxcommute}) is equivalent to the condition that  $a_2 \otimes b_1 = a_1 \otimes b_2$, which implies that  $b_1 - a_1 = b_2 - a_2$.

  \item Case 2: $a_1 > d_1$ and $a_2 = d_2$. Then we have $a_1 \otimes b_2 > b_2 \otimes d_1$. Hence (\ref{equation:Maxcommute}) is equivalent to the condition that  $a_2 \otimes b_1 = a_1 \otimes b_2 \oplus b_1 \otimes a_2$. Therefore, we have $a_1 \otimes b_2 \leq b_1 \otimes a_2$, i.e., $b_2 - a_2 \leq b_1 -a_1$.

  \item Case 3: $a_1 > d_1$ and $a_2 < d_2$. Then we have $a_1 \otimes b_2 > b_2 \otimes d_1$ and $a_2 \otimes b_1 < b_1 \otimes d_2$, which never implies  (\ref{equation:Maxcommute}).

  \item Case 4: $a_1 = d_1$ and $a_2 = d_2$. Then (\ref{equation:Maxcommute}) is equivalent to the condition that  $a_2 \otimes b_1 \oplus b_2 \otimes a_1 = a_1 \otimes b_2 \oplus b_1 \otimes a_2$, which is a tautology. 

  \item Case 5: $a_1 = d_1$ and $a_2 < d_2$. Since  $a_2 \otimes b_1 < b_1 \otimes d_2$,  (\ref{equation:Maxcommute}) is equivalent to the condition that  $b_2 \otimes d_1 = d_1 \otimes b_2 \oplus b_1 \otimes d_2$. Therefore, we have $d_1 \otimes b_2 \geq b_1 \otimes d_2$, i.e., $d_1 - b_1 \geq d_2 - b_2$.

  \item Case 6: $a_1 < d_1$ and $a_2 < d_2$. Since $a_1 \otimes b_2 < b_2 \otimes d_1$ and $a_2 \otimes b_1 < b_1 \otimes d_2$,  (\ref{equation:Maxcommute}) is equivalent to the condition that  $b_2 \otimes d_1 = b_1 \otimes d_2$. Therefore, we have $d_1 - b_1 = d_2 - b_2$.
\end{itemize}
This case analysis completes the proof. 
\end{proof}

\begin{lemma}
For an  $i \in [m]$,   let $N_i =\Mtt{a_i}{b_i}{\mathbb{0}}{d_i}$ with  
$a_i, b_i$, and $d_i$ in $\Br_{\max} \setminus \{\mathbb{0}\}$.
If
  \begin{enumerate}[(1)]
    \item $a_i > d_i$ for any $i \in [m]$ and there exists a constant $c \in \Br$ such that $c= b_i - a_i$ for any $i \in [m]$,  
     \item $a_i = d_i$ for any $i \in [m]$,
    \item[or]
    \item $a_i < d_i$ for any $i \in [m]$ and there exists a constant $c \in \Br$ such that $c = d_i - b_i$ for any $i \in [m]$,
  \end{enumerate}
  then the operation $\otimes$ is commutative on the set $\{N_1, \dots, N_m\}$.
\end{lemma}

The lemma means that the third entry of $\kappa$ is unnecessary, i.e., the lexicographic ordering for $\kappa^{\ast}$ gives an minimum permutation, where
\[
\kappa^{\ast}(N) =
\begin{cases}
    (-1,b-a) & (a>d),\\
    (0,0) & (a=d),\\
    (1,d-b) &(a<d)
\end{cases}
\]
for $N = \Mtt{a}{b}{\mathbb{0}}{d}$. Thus we obtain the next theorem.
\begin{theorem}
The problem (\ref{eq-maxtriangular}) for  $m=2$ can be solved by computing the lexicographic order for $\kappa^{\ast}$.
\end{theorem}
We remark that the permutation above is different from the one provided by Bouquard et al. \cite{MaxPlus}, which is obtained by the lexicographic ordering for $\kappa_{\text{BLB}}$, where
\[
\kappa_{\text{BLB}}(N) =
\begin{cases}
    (-1,b-a) & (a \geq d),\\
    (1,d-b) &(a<d).
\end{cases}
\]
By Lemma \ref{lemma:maxcommuting} (1), 
the ordering $\tilde{\kappa}$ obtained from the lexicographic ordering for $\kappa^{\ast}$ by sorting $N_i$'s with $a_i \geq d_i$ in nondecreasing order with respect to $b_i - a_i$ is also minimum. 
It is not difficult to see that $\tilde{\kappa}$ is  the lexicographic ordering for $\kappa_{\text{BLB}}$.
orders of optimal orderings of the minimization problem are hence optimal.

We also  note that it is difficult to introduce ``simultaneous triangularizability" in the max-plus algebra because only matrices that can be obtained from diagonal matrices by permuting the rows and/or columns are invertible.

Finally we prove Therem \ref{theorem:matrix-hardness1} (ii).

\begin{proof}[Proof of Therem \ref{theorem:matrix-hardness1} (ii)]
 We first show the case of $m=3$. 
 We reduce the strongly NP-hard problem \eqref{eq-maxtriangular} to the problem. 
    For an index $i \in [n]$, 
    let $N_i$ be a matrix in $ \Br_{\max}^{3\times3}$  such that 
    $$N_i=
    \begin{pmatrix*}
        a_i^{(1,1)} & a_i^{(1,2)} & a_i^{(1,3)}\\
        \mathbb{0} & a_i^{(2,2)} & a_i^{(2,3)}\\
        \mathbb{0} & \mathbb{0} & a_i^{(3,3)}
    \end{pmatrix*},
    $$
    where each $a_i^{(j,k)}\in\mathbb{Z}$.
    For a real $\gamma$ with $\gamma>\frac{n(n+1)}{2}$, we construct $n$ matrices 
    $M_1,\dots,M_n\in\Br^{3\times3}$ as follows:
    $$
    M_i = 
    \begin{pmatrix*}
        \gamma^{a_i^{(1,1)}} & \gamma^{a_i^{(1,2)}} & \gamma^{a_i^{(1,3)}}\\
        0 & \gamma^{a_i^{(2,2)}} & \gamma^{a_i^{(2,3)}}\\
        0 & 0 & \gamma^{a_i^{(3,3)}}
    \end{pmatrix*}.
    $$
    We note that the size of $M_i$ is 
    $O(\sum_{j,k} \log(\gamma^{a_i^{(j,k)}})) = O(\sum_{j,k} {a_i^{(j,k)}\log(\gamma)})$,
    which implies that the sum of the size of $M_1,\dots,M_n$ is bounded by a polynomial of $n$ and $\max_{i,j,k}a_i^{(j,k)}$ if $\gamma =\frac{n(n+1)}{2}+1$.
     In a manner similar to Lemma \ref{lemma:reduce_to_scheduling}, it is enough to show that a permutation $\sigma:[n]\rightarrow[n]$ that  minimizes
    $
    \begin{pmatrix*}
        1 & 0 & 0
    \end{pmatrix*}
    M^\sigma
    \begin{pmatrix*}
        0 \\ 0 \\ 1
    \end{pmatrix*}
    $
     also minimizes 
    $
    \begin{pmatrix*}
        \mathbb{1} & \mathbb{0} & \mathbb{0}
    \end{pmatrix*}
    \otimes
    N^\sigma
    \otimes
    \begin{pmatrix*}
        \mathbb{0} \\ \mathbb{0} \\ \mathbb{1}
    \end{pmatrix*}$. 
    For simplicity of notation, we write $v(M^\sigma)$ (resp., $v(N^\sigma)$) instead of 
    $
    \begin{pmatrix*}
        1 & 0 & 0
    \end{pmatrix*}
    M^\sigma
    \begin{pmatrix*}
        0 \\ 0 \\ 1
    \end{pmatrix*}
    $
    (resp.,
    $
    \begin{pmatrix*}
        \mathbb{1} & \mathbb{0} & \mathbb{0}
    \end{pmatrix*}
    \otimes
    N^\sigma
    \otimes
    \begin{pmatrix*}
        \mathbb{0} \\ \mathbb{0} \\ \mathbb{1}
    \end{pmatrix*}
    $).
    Then we need to show that,  for two permutations $\sigma$ and $\rho$, 
    $v(N^\sigma) < v(N^\rho)$ 
    implies 
    $v(M^\sigma) < v(M^\rho)$.
    By Lemma \ref{lemma:maxplus_linear}, we have 
    \begin{align*}
    v(M^\sigma)
    =
    \sum_{t\in V[\sigma]_{13}} \gamma^t \ \ \mbox{ and } \ \ 
    v(N^\sigma)
    =
    \max V[\sigma]_{13}. 
    \end{align*}
    Note that  
    \begin{align*}
        |V[\sigma]_{jk}|=
        \begin{cases}
            1 & \text{if } j=k,\\
            n & \text{if } j+1=k, \\
            \frac{n(n+1)}{2} & \text{if } j+2=k,
        \end{cases}
    \end{align*}
which can be easily proved by using the induction of $i \in [n]$. 
    Let $\sigma$ and $\rho$ be permutations such that  $v(N^\sigma)<v(N^\rho)$, then we have
    \begin{align*}
        v(M^\sigma)
        &\leq |V[\sigma]_{13}| \gamma^{v(N^\sigma)}
        = \frac{n(n+1)}{2} \gamma^{v(N^\sigma)}
        < \gamma^{v(N^\sigma)+1}
        \leq \gamma^{v(N^\rho)} 
        \leq v(M^\rho),
    \end{align*}
    which proves the case of $m=3$.

    For $m \geq 4$,  consider the problem to minimize the value
    \[
    \bm{w}^\top \bar{M}_{\sigma(n)}\cdots \bar{M}_{\sigma(1)} \bm{y},
    \]
    where 
    \[
    \bm{w}=\left(\begin{array}{c}
        \\
        O_{m-3,1}\\
        \\
        \hline
        1\\
        0\\
        0
    \end{array}\right), \,
    \bm{y}=\left(\begin{array}{c}
        \\
        O_{m-3,1}\\
        \\
        \hline
        0\\
        0\\
        1
        \end{array}\right),\, 
        \bar{M}_i= \left(\begin{array}{c|c}
        & \\
        O_{m-3,m-3} & O_{m-3,3}\\
        &\\
        \hline
        & \\
        O_{3, m-3} & M_i \\
        & 
        \end{array}\right).
    \]
    The symbol $O_{j,k}$ stands for the $j \times k$ zero matrix.  
It is clear that each $\bar{M_i}$ is nonnegative and upper triangular, and 
the objective value is equal to $v(M^\sigma)$, which proves the theorem.
\end{proof}

\bibliographystyle{abbrv}
\bibliography{reference.bib}
\end{document}